\documentclass{LMCS}

\usepackage{xspace}

\usepackage{gastex}
\usepackage{amssymb,amsmath,enumerate,hyperref}
\usepackage{xspace}
\renewcommand{\bar}[1]{\overline{#1}}
\newcommand{\cA}{{\mathcal A}}
\newcommand{\cB}{{\mathcal B}}
\newcommand{\Cc}{{\mathcal C}}
\newcommand{\Ch}{{\mathrm{Ch}}}

\newcommand{\cofinal}{{\mathrm{cofin}}}
\newcommand{\Col}{{\mathrm{Col}}}
\newcommand{\cP}{{\mathcal P}}
\newcommand{\cM}{{\mathcal M}}
\newcommand{\cH}{{\mathcal H}}
\newcommand{\cS}{{\mathcal S}}
\newcommand{\DLTL}{\ensuremath{\mathrm{PDL}}\xspace}

\newcommand{\iDLTL}{\ensuremath{\mathrm{iPDL}}\xspace}
\newcommand{\ind}{{\mathrm{ind}}}
\newcommand{\Inf}{{\mathrm{Inf}}}
\newcommand{\msg}{{\mathrm{msg}}}
\newcommand{\N}{{\mathbb N}}
\newcommand{\proc}{{\mathrm{proc}}}
\newcommand{\rA}{\mathrm{A}}
\newcommand{\rE}{\mathrm{E}}
\newcommand{\sub}{\mathrm{sub}}
\newcommand{\true}{{t\!t}}
\newcommand{\rein}{\mathrm{in}}
\newcommand{\raus}{\mathrm{out}}
\renewcommand{\phi}{\varphi}
\renewcommand{\uplus}{\mathrel{\mathaccent\cdot\cup}}

\newcommand{\cb}[1]{#1}

\def\doi{6 (3:16) 2010}
\lmcsheading%
{\doi}
{1--31}
{}
{}
{Jun.~\phantom05, 2009}
{Sep.~\phantom03, 2010}
{}   

\begin{document}

\title{Propositional Dynamic Logic for Message-Passing Systems}

\author[B.~Bollig]{Benedikt Bollig\rsuper a} 
\address{{\lsuper a}LSV, ENS Cachan, CNRS, France} 
\email{bollig@lsv.ens-cachan.fr} 
\thanks{{\lsuper a}Work was partly supported by the projects ANR-06-SETI-003
  DOTS and ARCUS {\^I}le de France-Inde.} 

\author[D.~Kuske]{Dietrich Kuske\rsuper b} 
\address{\cb{{\lsuper b}LaBRI Universit\'e de Bordeaux and CNRS, France}} 
\email{dietrich.kuske@tu-ilmenau.de} 

\author[I.~Meinecke]{Ingmar Meinecke\rsuper c} 
\address{{\lsuper c}Institut f\"ur Informatik, Universit\"at Leipzig, Germany} 
\email{meinecke@informatik.uni-leipzig.de} 
\thanks{{\lsuper c}Ingmar Meinecke was supported by the German Research Foundation (DFG)} 


\begin{abstract}
  We examine a bidirectional propositional dynamic logic (PDL) for finite and
  infinite message sequence charts (MSCs) extending LTL and TLC$^{-}$. By this
  kind of multi-modal logic we can express properties both in the entire
  future and in the past of an event. Path expressions strengthen the
  classical until operator of temporal logic. For every formula defining an
  MSC language, we construct a communicating finite-state machine (CFM)
  accepting the same language. The CFM obtained has size exponential in the
  size of the formula. This synthesis problem is solved in full generality,
  i.e., also for MSCs with unbounded channels. The model checking problem for
  CFMs and HMSCs turns out to be in PSPACE for existentially bounded MSCs. 
  Finally, we show that, for PDL with intersection, the semantics of a formula
  cannot be captured by a CFM anymore. 
\end{abstract}

\maketitle

\section{Introduction}

To make a system accessible to formal analysis and verification techniques, we
require it to be modeled mathematically. In this regard, automata-based models
have been widely used to describe the behavior of a system under
consideration. A natural model for finite processes that exchange messages via
FIFO-channels are communicating finite-state machines (CFMs) \cite{BraZ83}. In
a CFM, each process is modeled as a finite automaton that performs send and
receive actions and, in doing so, exchanges messages with other processes via
order-preserving communication channels. One single run of a CFM can be
described by a message sequence chart (MSC). MSCs are an important common
notation in telecommunication and are defined by an ITU
standard~\cite{mscitu96}. An MSC has both a formal definition and a
comprehensible visualization. 

Once we have an automata model $\cA$ of a system defining a set $L(\cA)$ of
possible executions, the next task might be to check if it satisfies a
requirements specification $\phi$, which represents a set $L(\phi)$ of
(desired) behaviors. Verification now amounts to the \emph{model-checking
  question}: do all possible behaviors of $\cA$ satisfy $\phi$, i.e., do we
have $L(\cA) \subseteq L(\phi)$? Many concrete instances of that problem have
been considered in the literature \cite{ClaGP00}. The original, and most
popular, ones are finite automata, Kripke structures, or B{\"u}chi automata as
system model, and temporal logics such as LTL \cite{Pnueli77} and CTL
\cite{ClarkeE81} as specification language. It is well-known that, for all
these choices, the corresponding model-checking problem is decidable. 

When we move to the setting of CFMs, which, due to a priori unbounded
channels, induce infinite-state systems, model checking becomes undecidable. 
Meenakshi and Ramanujam~\cite{MenR04} showed undecidability even for very
restrictive temporal logics (their results transfer easily from Lamport
diagrams to MSCs). One solution is to put a bound on the channel capacity. In
other words, the domain of behaviors is restricted to \emph{existentially}
$B$\emph{-bounded} MSCs, which can be executed without exceeding a fixed
channel bound~$B$. In~\cite{Pel00,MadM01,GenMSZ02,GenKM06}, the model-checking
problem was indeed tackled successfully for several logics by using this
restriction and following the automata-theoretic approach: (1) a formula
$\phi$ from a temporal logic or monadic second-order logic is translated into
a machine model $\cA_{\phi,B}$ that recognizes those models of $\phi$ that are
existentially $B$-bounded; (2) it is checked whether every existentially
$B$-bounded behavior of the system model is contained in the language of
$\cA_{\phi,B}$. 

\cb{ But on the other hand, we may apply temporal logic in the early stages of
  system development and start with specifying formulas to exemplify the
  intended interaction of the system to be. If so, we would like to synthesize
  a system model from a formula that captures precisely those behaviors that
  satisfy the formula. In other words, we ask whether a temporal-logic formula
  is realizable, i.e., whether the derived system is consistent and shows any
  reasonable behavior at all. Once a system is synthesized directly from its
  specification, it can be assumed to be correct a priori, provided the
  translation preserves the semantics of the specification.

  Though the assumption of bounded channels leads to the decidability of the
  model checking problem, it does not seem natural to restrict the channel
  size of the desired system in advance, especially when one is interested in
  the synthesis of a system from a specification. Despite the complexity of
  MSCs, we will provide in this paper a linear-time temporal logic for
  message-passing systems and solve its realizability problem in its full
  generality, i.e., under the assumption of a priori unbounded channels.

  Results from \cite{BolL06,BolK06} suggest to use an existential fragment of
  monadic second-order logic (EMSO) as a specification language. A formula
  from that fragment can be translated into a CFM that \emph{precisely}
  recognizes the models of the formula. This result holds without channel
  restriction. In this paper, we basically follow the approach from
  \cite{BolL06,BolK06}, but we propose a new logic: propositional dynamic
  logic (PDL) for MSCs. Our logic will prove useful for verification, as it is
  closed under negation and allows us to express interesting properties in an
  easy and intuitive manner. Like EMSO, but unlike full monadic second-order
  logic, every PDL formula $\phi$ can be effectively translated into a CFM
  $\cA_\varphi$ whose language is the set of models of $\phi$. This synthesis
  step is independent of any channel bound $B$ and would not become simpler if
  we took some $B$ into account. The size of the resulting CFM is exponential
  in the size of $\varphi$ and in the number of processes. Note that, by
  \cite{BolL06,BolK06}, EMSO is expressively equivalent to CFMs. Moreover, the
  set of CFM languages is not closed under complementation. As, on the other
  hand, \DLTL does not impose any restriction on the use of negation, we
  obtain that \DLTL is a proper fragment of EMSO although this is not obvious.

The model checking problem of CFMs againgst \DLTL formulas can be decided
in polynomial space for existentially $B$-bounded MSCs, following a standard procedure and using the translation of $\varphi$ into a CFM $\cA_\varphi$. Our PSPACE algorithm meets the lower bound that is imposed by
the complexity of LTL model checking for finite-state systems. 
We also show PSPACE completeness of the model checking problem of high-level MSCs (HMSCs) against \DLTL formulas (where the bound~$B$ is given implicitly by the HMSC). HMSCs are more abstract and
restrictive than CFMs, but can likewise be used as a model of a system. 
}

The final technical section considers an enriched logic: \iDLTL (\DLTL with
intersection). This extension seems natural to strengthen the expressive power
of the formulas. But adapting a proof technique from colored grids, we show
that \iDLTL is too strong for CFMs, i.e., there is an \iDLTL formula $\varphi$
such that no CFM accepts precisely the models of $\varphi$. 

\paragraph{\bf Related~Work.}~ For MSCs, there exist a few attempts to define
suitable temporal logics. Meenakshi and Ramanujam obtained exponential-time
decision procedures for several temporal logics over Lamport diagrams (which
are similar to MSCs)~\cite{MeeR00,MenR04}. Peled~\cite{Pel00} considered the
fragment $\mathrm{TLC^-}$ of the temporal logic $\mathrm{TLC}$ that was
introduced in~\cite{AlurPP95}. Like their logics, our logic is interpreted
directly over MSCs, not over linearizations; it combines elements from
\cite{MenR04} (global next operator, past operators) and \cite{Pel00} (global
next operator, existential interpretation of the until-operator). In
particular, however, it is inspired by \emph{dynamic} $\mathrm{LTL}$ as
introduced by Henriksen and Thiagarajan first for words~\cite{HenT99}. There,
standard $\mathrm{LTL}$ is extended by indexing the until operator with a
regular expression to make it more expressive. The same authors applied
dynamic $\mathrm{LTL}$ also to Mazurkiewicz traces but reasoned only about the
future of an event in the same process~\cite{HenT97}. In contrast, we might
argue about the whole future of an event rather than about one single process. 
Moreover, we provide past operators to judge about events that have already
been executed. We call our logic \DLTL because it is essentially the original
propositional dynamic logic as first defined by Fischer and
Ladner~\cite{FisL79} but here in the framework of MSCs. 
Although \DLTL can be seen as an extension of Peled's TLC$^-$, our decision
procedure is rather different. Instead of translating a \DLTL formula
$\varphi$ into a CFM directly, we use an inductive method inspired by
\cite{GasK03,GasK07,GasK10}. As $\mathrm{TLC^-}$ is a fragment of \DLTL, we actually
generalize the model checking result from \cite{Pel00}. 

\paragraph{\bf Outline.}~ In Section~\ref{sec:defs}, we define message
sequence charts, the logic \DLTL, and CFMs. We continue, in
Section~\ref{sec:cfmconstruct}, with several useful constructions for CFMs. 
Sections~\ref{sec:locform} and \ref{sec:globform} deal with the translation of
\DLTL formulas into CFMs. The model checking problem is tackled in
Section~\ref{sec:modelcheck} before we conclude, in Section~\ref{sec:ipdl},
with the result that \DLTL with intersection (\iDLTL) cannot be implemented in
terms of CFMs. 

\medskip

A preliminary version of this paper appeared as \cite{BKM-fsttcs07}.



\section{Definitions}\label{sec:defs}
The communication framework used in our paper is based on sequential processes
that exchange  messages asynchronously over point-to-point, error-free FIFO
channels. Let $\cP$ be a finite set of process identities which we fix
throughout this paper. Furthermore, let $\Ch=\{(p,q)\in\cP^2\mid p\neq q\}$
denote the set of \emph{channels}. Processes act by either sending a message,
that is denoted by $p!q$ meaning that process $p$ sends to process $q$, or by
receiving a message, that is denoted by $p?q$, meaning that process $p$
receives from process $q$. For any process $p\in\cP$, we define a local
alphabet (set of event types on $p$) $\Sigma_p=\{p!q,p?q\mid
q\in\cP\setminus\{p\}\}$, and we set $\Sigma=\bigcup_{p\in\cP}\Sigma_p$. 

\subsection{Message sequence charts}

A message sequence chart depicts processes as vertical lines, which are
interpreted as top-down time axes. Moreover, an arrow from one line to a
second corresponds to the communication events of sending and receiving a
message. Formally, message sequence charts are special labeled partial orders.
To define them, we need the following definitions: A \emph{$\Sigma$-labeled
  partial order} is a triple $M=(V,\le,\lambda)$ where $(V,\le)$ is a
partially ordered set and $\lambda:V\to\Sigma$ is a mapping. For $v\in V$ with
$\lambda(v)=p\theta q$ where $\theta\in\{!,?\}$, let $P(v)=p$ denote the
process that $v$ is located at. We set $V_p = P^{-1}(p)$. We define two binary
relations $\proc$ and $\msg$ on~$V$:
\begin{enumerate}[$\bullet$]
\item $(v,v')\in\proc$ iff $P(v)=P(v')$, $v<v'$, and, for any $u\in V$ with
  $P(v)=P(u)$ and $v\le u<v'$, we have $v=u$. The idea is that
  $(v,v')\in\proc$ whenever $v$ and $v'$ are two consecutive events
  of the same process. 
\item $(v,v')\in\msg$ iff there is a channel $(p,q)$ with
  $\lambda(v)=p!q$, $\lambda(v')=q?p$, and
  \[
  |\{u \mid \lambda(u)=p!q,~ u \le v\}|= |\{u \mid \lambda(u)=q?p,~ u
  \le v'\}|\ . 
  \]
  Here, the idea is that $v$ is a send event and $v'$ is the matching
  receive event. Since we model reliable FIFO-channels, this means
  that, for some $i$ and some channel $(p,q)$, $v$ is the $i^{th}$
  send and $v'$ the $i^{th}$ receive event on channel $(p,q)$. 
\end{enumerate}

\begin{defi}
  A \emph{message sequence chart} or \emph{MSC} for short is a
  $\Sigma$-labeled partial order $(V,\le,\lambda)$ such that
  \begin{enumerate}[$\bullet$]
  \item $\mathord{\le}=(\proc\cup\msg)^*$,
  \item $\{u\in V\mid u\le v\}$ is finite for any $v\in V$,
  \item $V_p$ is linearly ordered for any $p\in\cP$,
    and
  \item $|\lambda^{-1}(p!q)|=|\lambda^{-1}(q?p)|$ for any
    $(p,q)\in\Ch$. 
  \end{enumerate}
We refer to the elements of $V$ as \emph{events} or \emph{nodes}. 
\end{defi}

If $(V,\le,\lambda)$ is an MSC, then $\proc$ and $\msg$ are even
 injective partial functions, so $v'=\proc(v)$ as well as
$v=\proc^{-1}(v')$ are equivalent notions for $(v,v')\in\proc$;
$\msg(v)$ and $\msg^{-1}(v)$ are to be understood similarly. 

\cb{
An example MSC with three processes is pictured as a diagram in Figure~\ref{fig:exampleCFM}(b) on page~\pageref{fig:exampleCFM}. The processes are visualized as vertical lines going downwards and messages as horizontal directed edges between process lines.  
}
\subsection{Propositional dynamic logic}

\emph{Path expressions} $\pi$ and \emph{local formulas $\alpha$} are defined
by simultaneous induction. This induction is described by the following rules
\[
\begin{array}{rl}
  \pi ::= & \proc \mid \msg \mid \{\alpha\} \mid \pi;\pi 
  \mid \pi+\pi \mid \pi^*\\
  \alpha ::= & \true \mid \sigma \mid \alpha\lor\alpha \mid \neg\alpha
  \mid \left<\pi\right>\alpha \mid \left<\pi\right>^{-1}\alpha
\end{array}
\] where $\sigma$ ranges over the alphabet $\Sigma$. 

Local formulas express properties of single nodes in MSCs. To define the
semantics of local formulas, let therefore $M=(V,\le,\lambda)$ be an MSC and
$v$ a node from~$M$. Then we define
\begin{align*}
  M,v\models\sigma &\iff \lambda(v)=\sigma\text{~~ for }\sigma\in\Sigma\\
  M,v\models\alpha_1\lor\alpha_2
  &\iff M,v\models\alpha_1\text{ or }M,v\models\alpha_2\\
  M,v\models\neg\alpha
  &\iff M,v\not\models\alpha\\
  \intertext{The idea of forward-path modalities $\left<\pi\right>\alpha$
  indexed by a path expression $\pi$ is to perform a program $\pi$ and then to
  check whether $\alpha$ is satisfied. Thereby, $\pi$ is a rational expression
  over $\proc$ and $\msg$ describing paths in a MSC but allows also for tests
   $\{\alpha\}$ in following the paths defined by $\pi$. Formally, the semantics of
  \emph{forward}-path formulas $\left<\pi\right>\alpha$ is given by}
  M,v\models\left<\proc\right>\alpha &\iff \text{there exists }v'\in
  V\text{ with }(v,v')\in\proc
  \text{ and } M,v'\models\alpha\\
  M,v\models\left<\msg\right>\alpha &\iff \text{there exists }v'\in
  V\text{ with }(v,v')\in\msg
  \text{ and } M,v'\models\alpha\\
  M,v\models\left<\{\alpha\}\right>\beta
  &\iff M,v\models\alpha\text{ and }M,v\models\beta\\
  M,v\models\left<\pi_1;\pi_2\right>\alpha
  &\iff M,v\models\left<\pi_1\right>\left<\pi_2\right>\alpha\\
  M,v\models\left<\pi_1+\pi_2\right>\alpha
  &\iff M,v\models\left<\pi_1\right>\alpha\lor\left<\pi_2\right>\alpha\\
  M,v\models\left<\pi^*\right>\alpha &\iff 
         \text{there exists $n\ge 0$
         with } M,v\models(\left<\pi\right>)^n\alpha\\
  \intertext{The semantics of \emph{backward}-path formulas 
    $\left<\pi\right>^{-1}\alpha$ is defined similarly:}
  M,v\models\left<\proc\right>^{-1}\alpha &\iff \text{there exists
  }v'\in V\text{ with }(v',v)\in\proc
  \text{ and } M,v'\models\alpha\\
  M,v\models\left<\msg\right>^{-1}\alpha &\iff \text{there exists
  }v'\in V\text{ with }(v',v)\in\msg
  \text{ and } M,v'\models\alpha\\
  M,v\models\left<\{\alpha\}\right>^{-1}\beta
  &\iff M,v\models\alpha\text{ and }M,v\models\beta\\
  M,v\models\left<\pi_1;\pi_2\right>^{-1}\alpha &\iff
  M,v\models\left<\pi_1\right>^{-1}\left<\pi_2\right>^{-1}\alpha\\
  M,v\models\left<\pi_1+\pi_2\right>^{-1}\alpha &\iff
  M,v\models\left<\pi_1\right>^{-1}\alpha
  \lor\left<\pi_2\right>^{-1}\alpha\\
  M,v\models\left<\pi^*\right>^{-1}\alpha &\iff
    \text{there exists $n\ge 0$ with }
     M,v\models(\left<\pi\right>^{-1})^n\alpha
\end{align*}

Semantically, a local formula of the form $\langle (\{\alpha\} ; (\proc +
\msg))^\ast \rangle \beta$ corresponds to the until construct $\alpha
\mathcal{U} \beta$ in Peled's TLC$^-$ \cite{Pel00}. In TLC$^-$, however, one cannot express
properties such as ``there is an even number of messages from $p$ to $q$'',
which is easily expressible in PDL. 

Global properties of an MSC are Boolean combinations of properties of
the form ``there exists a node satisfying the local
formula~$\alpha$''. These global properties are expressed by
\emph{global formulas~$\varphi$} whose syntax is given by
\[
\begin{array}{rl}
  \varphi ::= & \rE\alpha \mid \rA\alpha \mid \varphi\lor\varphi \mid \varphi\land\varphi
\end{array}
\]
where $\alpha$ ranges over the set of local formulas. The semantics is
defined by
\begin{align*}
  M\models\rE\alpha &\iff
  \text{ there exists a node $v$ with }M,v\models\alpha\\
  M\models\rA\alpha &\iff
  M,v\models\alpha\text{ for all nodes $v$ }\\
  M\models\varphi_1\lor\varphi_2
  &\iff M\models\varphi_1\text{ or }M\models\varphi_2\\
  M\models\varphi_1\land\varphi_2 &\iff M\models\varphi_1\text{ and
  }M\models\varphi_2
\end{align*}
Note that our syntax of global formulas does not allow explicit
negation. But since we allow existential and universal quantification
as well as disjunction and conjunction, the expressible properties are
closed under negation. 
\cb{
\begin{exa}\label{exam_formula}
 For $i\in\cP$, we put $P_i = \bigvee_{j\in\cP, j\neq i}  (i!j\vee i?j)$, i.e., $M,v\models P_i$ iff $P(v)=i$ for every MSC $M=(V,\le,\lambda)$ and $v\in V$. Now the global formula
  \begin{align*}
    \varphi_{i,j}^{=2}=\rA\bigl(P_i \longrightarrow (\left<\proc^*;\msg;\proc^*;\msg\right>P_j)\bigr)
  \end{align*}
states that process $j$ can always be reached from process $i$ with exactly two messages (using an intermediate process in between).
\end{exa}
}

\begin{defi}\label{D-subformula}
  The set of subformulas $\sub(\alpha)$ of a local formula $\alpha$ and the
  set of subformulas $\sub(\pi)$ of a path expression $\pi$ are defined by
  synchronous induction as follows:
  \begin{align*}
    \sub(\proc)=\sub(\msg)&=\emptyset\\
    \sub(\{\alpha\})&=\sub(\alpha)\\
    \sub(\pi_1;\pi_2)=\sub(\pi_1+\pi_2)&=\sub(\pi_1)\cup\sub(\pi_2)\\
    \sub(\pi^*)&=\sub(\pi)\\
    \intertext{ and } \sub(\sigma)&=\{\sigma\} \text{ for }\sigma\in\Sigma\\
    \sub(\neg\alpha)&=\{\neg\alpha\}\cup\sub(\alpha)\\
    \sub(\alpha\lor\beta)&=\{\alpha\lor\beta\}\cup\sub(\alpha)\cup\sub(\beta)\\
    \sub(\left<\pi\right>\alpha)&=\{\left<\pi\right>\alpha\}\cup\sub(\pi)\cup\sub(\alpha)\\
    \sub(\left<\pi\right>^{-1}\alpha)&=\{\left<\pi\right>^{-1}\alpha\}\cup\sub(\pi)\cup\sub(\alpha)
  \end{align*}
\end{defi}
Thus, in addition to the obvious definition, a subformula of a path expression
is any of the local formulas occurring in the path expression as well as any
subformula of these local formulas. In particular, contrary to what one might
expect, a rather long local formula like $\varphi =
\left<\proc;\{\sigma\};\proc;\{\sigma\};\proc;\{\sigma\};\proc;\{\sigma\}\right>\sigma$
has only two subformulas, namely $\varphi$ itself and $\sigma$. The number of
subformulas of $\alpha$ is bounded by the length of $\alpha$, but the length
of $\alpha$ cannot be bounded in terms of the number of subformulas. 

Note that a path expression $\pi$ is a regular expression over the following alphabet 
$\{\proc,\msg,\{\alpha_1\},\dots,\{\alpha_n\}\}$ for some local formulas
$\alpha_i$. The \emph{size $s(\pi)$} of $\pi$ is defined by
$s(\{\alpha\})=s(\proc)=s(\msg)=1$,
$s(\pi_1+\pi_2)=s(\pi_1;\pi_2)=s(\pi_1)+s(\pi_2)$ and $s(\pi^*)=s(\pi)$ (i.e.,
it is the number of occurrences of $\{\alpha\}$, $\msg$, and $\proc$ in the
regular expression $\pi$). Note that the size of the path expression
$\{\alpha\}$ is~$1$, independent from the concrete form of the local
formula~$\alpha$.

\subsection{Communicating finite-state machines}

One formalism to describe (asynchronous) communication
protocols are \emph{communicating finite-state machines} (CFM for
short)~\cite{BraZ83}. They form a basic model for distributed
algorithms based on asynchronous message passing between concurrent
processes. Thus, the basic actions performed are just sending and
receiving of messages (i.e., letters from $\Sigma$). 

\cb{
A CFM $\cA$ consists of a collection of finite automata $\cA_p$, one for each process $p\in\cP$. The automaton $\cA_p$ performs the actions of process $p$, i.e., the send events $p!q$ and the receive events $p?q$ for all $q\neq p$. Moreover, the single automata synchronize by control messages from some finite set $C$. Whenever $\cA_p$ sends a message to $\cA_q$, then $\cA_p$ and $\cA_q$ share some common control message $c\in C$. The final states of a CFM are defined globally and the local components of a final state have either to be repeated infinitely often by the process or the process terminates in such a local state. 
}

We extend the
alphabet $\Sigma$ for later purposes to $\Sigma\times\{0,1\}^n$ for some
$n\in\N$ -- the classical model is obtained by setting $n=0$ in the
below definition (in which case we write $\cA=(C,(\cA_p)_{p \in \cP},F)$.) 

\begin{defi}
  A \emph{communicating finite-state machine} (or, simply, \emph{CFM})
  is a structure $\cA=(C,n,(\cA_p)_{p \in \cP},F)$ with $n\in\N$ where
  \begin{enumerate}[$\bullet$]
  \item $C$ is a finite set of \emph{message contents} or
    \emph{control messages},
  \item $\cA_p=(S_p,\rightarrow_p,\iota_p)$ is a finite labeled
    transition system over the alphabet $\Sigma_p\times\{0,1\}^n\times
    C$ for any $p\in\cP$ (i.e., ${\rightarrow_p}\subseteq
    S_p\times(\Sigma_p\times\{0,1\}^n\times C)\times S_p$) with
    initial state $\iota_p\in S_p$,
  \item $F\subseteq\prod_{p\in\cP}S_p$ is a set of global final
    states. 
  \end{enumerate}
\end{defi}

Now let $\cA$ be a CFM as above, $M=(V,\le,\lambda)$ be an MSC, and
$c:V\to\{0,1\}^n$. A \emph{run of $\cA$ on $(M,c)$} is a pair $(\rho,\mu)$ of mappings
$\rho:V\to\bigcup_{p\in\cP}S_p$ and $\mu:V\to C$ such that, for any $v\in V$,
\begin{enumerate}[(1)]
\item $\mu(v)=\mu(\msg(v))$ if $\msg(v)$ is defined,
\item $(\rho(\proc^{-1}(v)),\lambda(v),c(v),\mu(v),\rho(v))\in
  {\to_{P(v)}}$ if $\proc^{-1}(v)$ is defined, and\\
  $(\iota_p,\lambda(v),c(v),\mu(v),\rho(v))\in {\to_{P(v)}}$
  otherwise. 
\end{enumerate}
In order to define when the run $(\rho,\mu)$ is accepting, we will use
B\"uchi-conditions on each process. For this, one is usually interested in the
set of states that appear infinitely often. But since, even in an infinite
MSC, some of the processes may execute only finitely many events, the set of
states appearing infinitely often is here generalized to the set of states
that appear \emph{cofinally}: Let $\cofinal_\rho(p)=\{s\in S_p\mid\forall v\in
V_p~\exists v'\in V_p:v\le v'\land \rho(v')=s\}$. Then the run $(\rho,\mu)$ is
\emph{accepting} if there is some $(s_p)_{p\in\cP}\in F$ such that
$s_p\in\cofinal_\rho(p)$ for all $p\in\cP$. The \emph{language} of $\cA$ is
the set $L(\cA)$ of all pairs $(M,c)$ that admit an accepting run.

\cb{

\newcommand{\AHLMSC}{1.8}
\newcommand{\AHlMSC}{1.5}
\newcommand{\AHaMSC}{30}

\begin{figure}
\begin{tabular}{ccc}
\mbox{
\begin{picture}(80,60)(-20,-5)
\unitlength=0.3em
 \gasset{Nframe=y,Nw=5,Nh=5,Nmr=8,ilength=4}
\gasset{AHangle=25,AHLength=1.3,AHlength=1.2}

 \node[Nmarks=i,iangle=90,ilength=3](s0)(-7,30){$s_0$}
 \node(s1)(-7,15){$s_1$}
 \node(s2)(-7,0){$s_2$}
 \node[Nmarks=r](s3)(-22,15){$s_3$}

 \node[Nmarks=ir,iangle=90,ilength=3](t0)(18,30){$t_0$}
 \node(t1)(18,15){$t_1$}

 \node[Nmarks=ir,iangle=90,ilength=3](q0)(40,30){$q_0$}

 \drawedge[curvedepth=-2,ELside=r](s0,s1){\small 1!2, $\mathrm{r}$~}
 \drawedge[ELside=l](s1,s2){\small 1?2, $\checkmark$}
 \drawedge[curvedepth=-2,ELside=r](s1,s0){\small 1?2, $\mathrm{x}$}
 \drawedge[ELside=l](s2,s3){\small 1!3, $\mathrm{c}$}
 \drawedge[ELside=l](s3,s1){\small 1!2, $\mathrm{r}$}

 \drawedge[curvedepth=-2,ELside=r](t0,t1){\small 2?1, $\mathrm{r}$}
 \drawedge[curvedepth=-2,ELside=r](t1,t0){\!\!\begin{tabular}{l}{\small 2!1,
       $\checkmark$}\vspace{-6ex}\\{\small 2!1,
       $\mathrm{x}$}\end{tabular}}

 \drawloop[loopdiam=5,loopangle=270](q0){\small 3?1, $\mathrm{c}$}

 \gasset{Nframe=n,Nadjust=w,Nh=0,Nmr=0}
 \node(A1)(-7,39){$\cA_{\small \mathrm{Client}}$}
 \node(A2)(16,39){$\cA_{\small \mathrm{Server}}$}
 \node(A3)(39,39){$\cA_{\small \mathrm{Interface}}$}
\end{picture}
}
&
&
\mbox{
\begin{picture}(30,60)(-5,-15)
\unitlength=0.2em 

  \node[Nw=16,Nh=5,Nmr=0](P)(0,40){\tiny $\mathrm{Client} (1)$}
  \node[Nw=16,Nh=5,Nmr=0](Q)(20,40){\tiny $\mathrm{Server} (2)$}
  \node[Nw=20,Nh=5,Nmr=0](R)(40,40){\tiny $\mathrm{Interface} (3)$}

  \node[Nframe=n,Nh=3](PP)(0,-10){{$\vdots$}} 
  \node[Nframe=n,Nh=3](QQ)(20,-10){{$\vdots$}}
  \node[Nframe=n,Nh=3](RR)(40,-10){{$\vdots$}}

  \node[Nframe=n,Nw=12,Nh=5,Nmr=0](P)(0,40){}
  \node[Nframe=n,Nw=12,Nh=5,Nmr=0](Q)(20,40){}
  \node[Nframe=n,Nw=12,Nh=5,Nmr=0](R)(40,40){}

  \node[Nframe=n,Nh=3](PP)(0,-10){} 
  \node[Nframe=n,Nh=3](QQ)(20,-10){}
  \node[Nframe=n,Nh=3](RR)(40,-10){}

  \drawedge[AHangle=0](P,PP){}
  \drawedge[AHangle=0](Q,QQ){}
  \drawedge[AHangle=0](R,RR){}

  \gasset{Nfill=y,Nw=2,Nh=2,Nframe=y}
  \node(P1)(0,30){}
  \node(Q1)(20,30){}
  \drawedge[AHangle=\AHaMSC,AHLength=\AHLMSC,AHlength=\AHlMSC](P1,Q1){}

  \node(P1)(0,22){}
  \node(Q1)(20,22){}
  \drawedge[AHangle=\AHaMSC,AHLength=\AHLMSC,AHlength=\AHlMSC,ELside=r](Q1,P1){}

  \node(P1)(0,14){}
  \node(Q1)(20,14){}
  \drawedge[AHangle=\AHaMSC,AHLength=\AHLMSC,AHlength=\AHlMSC](P1,Q1){}

  \node(P1)(0,6){}
  \node(Q1)(20,6){}
  \drawedge[AHangle=\AHaMSC,AHLength=\AHLMSC,AHlength=\AHlMSC,ELside=r](Q1,P1){}

  \node(P1)(0,-2){}
  \node(Q1)(40,-2){}
  \drawedge[AHangle=\AHaMSC,AHLength=\AHLMSC,AHlength=\AHlMSC,ELside=l](P1,Q1){}

\node[Nframe=n,Nfill=n,Nw=0,Nh=0](P)(10,-6){{$\vdots$}}
\node[Nframe=n,Nfill=n,Nw=0,Nh=0](P)(30,-6){{$\vdots$}}
\end{picture}
}\\
(a) A CFM over $\{\mathrm{Client},\mathrm{Server},\mathrm{Interface}\}$. & & (b) An infinite MSC.
\end{tabular}
\caption{A CFM and an infinite MSC accepted by it.\label{fig:exampleCFM}}
\end{figure}
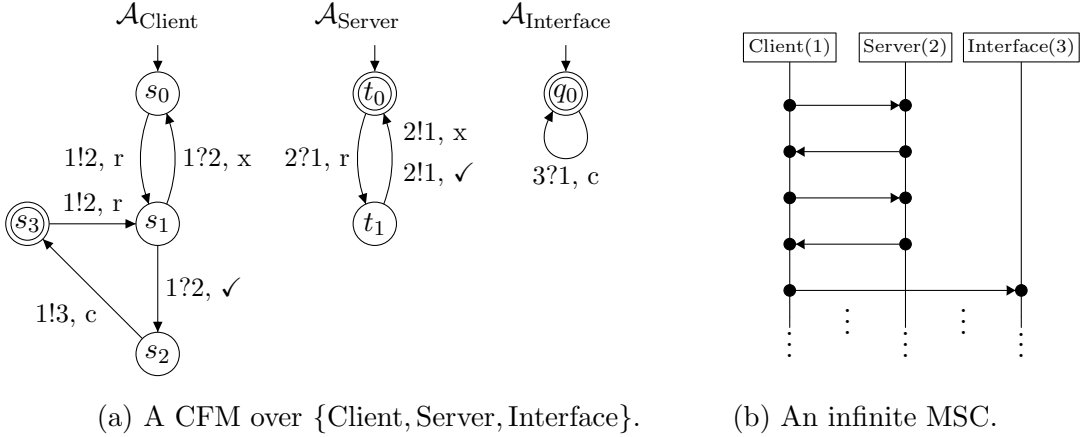

}

\cb{
  \begin{exa}
    Consider the CFM illustrated in Figure~\ref{fig:exampleCFM}(a). A client
    (process $1$) communicates with a server (process $2$) sending requests
    (message content r) to receive permission to send a message to the
    interface (process $3$). If the server refuses permission (message content
    x), the request is repeated. But if permission is given (message content
    $\checkmark$), then the client sends a message~$c$ to the interface. Now
    the client can start again to send requests to the server. Here, the only
    accepting state is $(s_3,t_0,q_0)$. Thus the client can either stop after
    sending a message to the interface (and all other processes also stop) or
    the client has to send infinitely many messages to the interface, i.e.,
    every request of the client is eventually followed by a communication with
    the interface.

    The MSC pictured in Figure~\ref{fig:exampleCFM}(b) is one possible
    behavior of the CFM. Moreover, any MSC $M$ accepted by this CFM satisfies
    the formula $\varphi_{2,3}^{=2}$ from Example~\ref{exam_formula}.
  \end{exa}
}

\section{Constructions of CFMs}\label{sec:cfmconstruct}

In this section, we present some particular CFMs and constructions of
CFMs. The purpose is twofold: the results will be used later, and the
reader shall become acquainted with the computational power of
CFMs. Hence, some readers might choose to skip the details of this section in a first
reading. 

\subsection{Intersection}
Here, we show that the intersection of languages accepted by CFMs can
again be accepted by a CFM. Since the acceptance by a CFM is defined
in terms of a B\"uchi-condition, we can adopt the flag
construction~\cite{Cho74} from the theory of word automata with
B\"uchi-acceptance condition (cf.~proof of Lemma~1.2
in~\cite{Tho90a}). The additional problem we face here is the
interplay between different processes. 

The basic idea of our construction is as follows (for two CFMs $\cA^1$
and $\cA^2$): each local process guesses an accepting global state
$f^1=(f^1_p)_{p\in\cP}$ of $\cA^1$ and $f^2=(f^2_p)_{p\in\cP}$ of
$\cA^2$. Then, locally, process $p$ simulates both CFMs $\cA^1$ and
$\cA^2$ and checks that $f^1_p$ and $f^2_p$ are visited infinitely
often (here, one relies on Choueka's flag construction). Hence, the
set of local states of process $p$ equals $F^1\times F^2\times
S^1_p\times S^2_p\times\{0,1,2\}$. A global state is accepting if all
the guesses locally made coincide and if the local processes accept
according to the flag construction. 

Recall that the set of accepting states $F^1$ is a set of tuples, its
maximal size is therefore $\prod_{p\in\cP}|S^1_p|$. Thus, the
intersection of two CFMs with $s$ local states per process can result
in a CFM with $s^{2|\cP|}\cdot s^2\cdot 3=s^{O(|\cP|)}$ many local
states per process. 

Now suppose that $F^1$ and $F^2$ are  direct products, i.e.,
$F^1=\prod_{p\in\cP}F^1_p$ for some sets $F^1_p\subseteq S^1_p$ and,
similarly, $F^2=\prod_{p\in\cP}F^2_p$ for some sets $F^2_p\subseteq
S^2_p$. Then, in the above construction, it is not necessary for the
local guesses to coincide -- which makes them superfluous (cf.\ Proof
of Lemma \ref{L-intersection-index=1} below). Thus, in this case, the
set of local states of process $p$ will just be $S^1_p\times
S^2_p\times\{0,1,2\}$, in particular, it will not be exponential in
the number of processes. 

To use this simplification of the construction, we introduce the
following notion. 

\begin{defi}\label{D-index}
  Let $F\subseteq \prod_{p\in\cP}S_p$. The \emph{index of $F$} is the
  least number $n$ such that there are sets $F_p^i\subseteq S_p$ for
  $p\in\cP$ and $1\le i\le n$ with $F=\bigcup_{1\le i\le
    n}\prod_{p\in\cP} F_p^i$. 

  The \emph{index of a CFM} is the index of its set of accepting
  states. 
\end{defi}

Clearly, the index of a CFM is bounded by $s^{|P|}$ where $s$ is the
maximal size of a set of local states $S_p$. To see that it can indeed
be quite large, let $\cP=S_p=[n]$ for all $p\in\cP$ (where we let
$[n]=\{1,\ldots,n\}$).  Furthermore, let $F$ be the set of all tuples
$(s_p)_{p\in\cP}$ such that $\{s_p\mid p\in\cP\}=[n]$, i.e., the set
of surjections from $[n]$ onto $[n]$.  Hence $F$ contains $n!$ many
elements. Any two of them differ in at least two positions. Hence the
index of $F$ equals its size and is exponential in $n$ and therefore
in $|P|$. Despite this exponential example, we will encounter only
small indices in our constructions.

\begin{lem}\label{L-intersection-index=1}
  For $1\le i\le m$, let
  $\cA^i=(C^i,n,(S_p^i,\to_p^i,\iota_p^i)_{p\in\cP},F^i)$ be CFMs of
  index $1$. Then there exists a CFM $\cA$ of index $1$ that accepts
  $(M,c)$ with $M$ an MSC and $c:V\to\{0,1\}^n$ iff it is accepted by
  $\cA^i$ for all $i\in[m]$. 

  The set of messages of~$\cA$ is $\prod_{i\in[m]}C^i$ and the set of local
  states of process $p$ is $\{0,1,\dots,m\}\times\prod_{i\in[m]}S_p^i$. 
\end{lem}

\begin{proof}
  Since $F^i$ has index~$1$, there exist sets $F_p^i\subseteq S_p^i$
  with $F^i=\prod_{p\in\cP}F_p^i$. 

  The idea of the proof is that $\cA$ will simulate all the machines $\cA^i$
  in parallel. In addition, it checks that, for each $p\in\cP$ and $i\in[m]$,
  some state from $F_p^i$ is assumed cofinally (i.e., infinitely often or, if
  process $p$ executes only finitely many events, at the last event from $p$). 
  Formally, we set $\iota_p=
  \begin{cases}
    (m,\iota_p^1,\dots,\iota_p^m) &\text{ if
    }(\iota_p^1,\dots,\iota_p^m)\in\prod_{i\in[m]}F_p^i\\
    (0,\iota_p^1,\dots,\iota_p^m) &\text{ otherwise}
  \end{cases}
  $ and\linebreak $F=\prod_{p\in\cP}\left(\{m\}\times\prod_{i\in[m]} S_p^i
  \right)$.  Furthermore,
  $(a,(s_i)_{i\in[m]})\xrightarrow{\sigma,c,(b_i)_{i\in[m]}}_p
  (a',(s'_i)_{i\in[m]})$ with $b_i\in C^i$ is a transition of $\cA$
  iff
  \begin{enumerate}[(1)]
  \item $s_i\xrightarrow{\sigma,c,b_i}{}^i_p s'_i$ is a transition of
    $\cA^i$ for all $i\in[m]$
  \item $a'=
    \begin{cases}
      m & \text{ if }s_i\in F_p^i\text{ for all }i\in[m]\\
      0 & \text{ if }a=m\text{ and }
      s_i\notin F_p^i\text{ for some }i\in[m]\\
      a+1 & \text{ if }a<m, s_{a+1}\in F_p^{a+1}\text{ and }
      s_i\notin F_p^i\text{ for some }i\in[m]\\
      a & \text{ otherwise.} 
    \end{cases}
    $
  \end{enumerate}
  Recall the classical flag construction for $\omega$-word automata. 
  There, the value of the counter~$a$ indicates that the composite
  machine waits for an accepting state of the simulated machine~$a+1$;
  a value $m$ indicates that all simulated machines went through some
  accepting states. Here, we do the same. But, in addition, if all
  component states of the composite machine are accepting, then we set
  the counter value directly to~$m$. This is useful when process~$p$
  executes only finitely many events. Then, at its final event~$v$,
  all the component machines have to be in some accepting state. For
  processes executing infinitely many events, this is of no
  importance. 
\end{proof}

\begin{prop}\label{P-intersection-index-large}
  For $i\in[m]$, let
  $\cA^i=(C^i,n,(S_p^i,\to_p^i,\iota_p^i)_{p\in\cP},F^i)$ be a CFM of
  index $\ell_i$. Then there exists a CFM $\cA$ of index $\prod_{1\le
    i\le m}\ell_i$ that accepts $(M,c)$ with $M$ an MSC and
  $c:V\to\{0,1\}^n$ iff it is accepted by $\cA^i$ for all $i\in[m]$. 

  The set of messages of~$\cA$ is $\prod_{i\in[m]}C^i$. Moreover, the set of
  local states of process~$p$ is $\{\iota_p\}\uplus(\{0,1,2,\dots,m\}\times
  \prod_{i\in[m]}S_p^i\times\prod_{i\in[m]}[\ell_i])$. 
\end{prop}

\begin{proof}
  Since the index of $\cA^i$ is $\ell_i$, its language is the union of
  languages $L_1^i,\dots,L_{\ell_i}^i$ that can each be accepted by a CFM of
  index~$1$. The language in question is therefore given by
  $\bigcup_{j\in\prod_{i\in[m]}[\ell_i]}\bigcap_{i\in[m]} L_{j_i}^i$. By
  Lemma~\ref{L-intersection-index=1}, the intersection $\bigcap_{i\in[m]}
  L_{j_i}^i$ can be accepted by a CFM of index~$1$ with set of local states
  $\{0,1,2,\dots,m\}\times\prod_{i\in[m]} S_p^i\times\{j\}$. The disjoint
  union of all these CFMs (together with new local initial states) accepts the
  language in question; its set of local states equals
  $\{\iota_p\}\uplus(\{0,1,2,\dots,m\}\times\prod_{i\in[m]}
  S_p^i\times\prod_{i\in[m]}[\ell_i])$ and its index is $\prod_{1\le
    i\le m}\ell_i$ as claimed. 
\end{proof}

\subsection{Infinitely running processes}
For an MSC $M=(V,\le,\lambda)$, let $\Inf(M)\subseteq\Ch$ denote the
set of those channels $(p,q)$ that are used infinitely often, i.e.,
$\Inf(M)=\{(p,q)\in\Ch\mid\lambda^{-1}(p!q)\text{ is infinite}\}$. 
{}From a set $I\subseteq\Ch$, we want to construct a CFM of index~$1$
that checks whether $\Inf(M)=I$. 

\begin{lem}\label{L-finite-channels}
  Let $I\subseteq\Ch$. There exists a CFM~$\cA_1$ of index~$1$ with
  three local states per process and one message that accepts an MSC
  $M$ iff $\Inf(M)\subseteq I$. 
\end{lem}

\begin{proof}
  The sets of local states are given by $S_p=\{0,1,2\}$ for any
  $p\in\cP$, the state $0$ is locally initial. The only control
  message is $1$. Then we set $a\xrightarrow{\sigma,1}_p b$ iff ($a=b$
  and $\sigma$ uses a channel from $I$) or ($a<b$ and $\sigma$ does
  not use a channel from $I$) or $a=b=1$. Then state $0$ indicates
  that no channel of $\Ch\setminus I$ has been used, $1$ indicates
  that some channel from $\Ch\setminus I$ has been used and that some
  channel will be used, and $2$ denotes that some channel from
  $\Ch\setminus I$ has been used but none will ever be used in the
  future. Hence, process $p$ uses the channels from $\Ch\setminus I$
  only finitely often iff it can visit $0$ or $2$ cofinally. Setting
  $F=\prod_{p\in\cP}\{0,2\}$ therefore finishes the construction of
  the desired CFM. 
\end{proof}

\begin{lem}\label{L-infinite-channels}
  Let $I\subseteq\Ch$. There exists a CFM~$\cB_1$ of index~$1$ with
  $4^{|\cP|}$ local states per process and one message that accepts an
  MSC $M$ iff $I\subseteq\Inf(M)$. 
\end{lem}

\begin{proof}
  For $p\in\cP$ let $S_p=\{0,1\}^{\Sigma_p}$ and set $C=\{1\}$. The
  locally initial state $\iota_p\in S_p$ maps all $\tau\in\Sigma_p$ to
  $0$. Then we set $g\xrightarrow{\sigma,1}g'$ for $g,g'\in S_p$ and
  $\sigma\in\Sigma_p$ iff $g'(\tau)=
  \begin{cases}
    g(\tau) & \text{ if }\tau\neq\sigma\\
    1-g(\tau) & \text{ otherwise}
  \end{cases}
  $ for all $\tau\in\Sigma_p$.  Thus, the local process $p$ counts
  modulo~$2$ the number of occurrences of any local action. The
  channel $(p,q)$ is used infinitely often iff the following two
  properties hold:
  \begin{enumerate}[$\bullet$]
  \item Process $p$ visits a state $g_p$ with $g_p(p!q)=0$ cofinally. 
  \item Process $q$ visits a state $g_q$ with $g_q(q?p)=1$ cofinally. 
  \end{enumerate}
  Therefore, a global state $(g_p)_{p\in\cP}$ is final (i.e., belongs
  to $F$) iff, for any $(p,q)\in I$, we have $g_p(p!q)=0$ and
  $g_q(q?p)=1$. 
\end{proof}

\begin{prop}\label{P-(in)finite-channels}
  Let $I\subseteq\Ch$. There exists a CFM~$\cB$ of index~$1$ with
  $3\cdot3\cdot4^{|\cP|}$ local states per process and one message
  that accepts an MSC $M$ iff $I=\Inf(M)$. 
\end{prop}

\begin{proof}
  Follows immediately from Lemmas~\ref{L-finite-channels},
  \ref{L-infinite-channels}, and \ref{L-intersection-index=1}. 
\end{proof}

\subsection{The color language}
\label{sec:color-language}

In this section, we build a CFM that accepts some ``black/white
colored'' MSCs. \cb{The aim is that whenever a coloring is accepted, then any infinite
path in the MSC has infinitely many color changes (cf.\
Cor.~\ref{C-inf-many-color-changes}).} This language will be the
crucial ingredient in our handling of forward-path formulas of the
form~$\left<\pi\right>\alpha$ (cf.~Section~\ref{ssec:forward}). 

For the time being, we proceed as follows: first, we define a language
$\Col$ whose elements are colored
MSCs~$(M,c)$. Prop.~\ref{P-Col-regular} shows that this language can
be accepted by a CFM. Cor.~\ref{C-inf-many-color-changes} ensures that
any infinite path in $(M,c)\in\Col$ has infinitely many color
changes. We do not prove the converse (which is actually false), but
will see later that sufficiently many colorings with this property
belong to $\Col$ (Lemma~\ref{L-hhh}). 

Let $M$ be an MSC and $c:V\to\{0,1\}$. On $V$, we define an
equivalence relation~$\sim$ setting $u\sim w$ iff $P(u)=P(w)$ and, for
all $v\in V$ with ($u\le v\le w$ or $w\le v\le u$) and $P(u)=P(v)$, we
have $c(u)=c(v)=c(w)$ (i.e., a $\sim$-equivalence class is a maximal
monochromatic interval on a process line). 

Let $\Col$ be the set of all pairs $(M,c)$ with $c:V\to\{0,1\}$ such that the
following hold
\begin{enumerate}[(1)]
\item if $v$ is minimal on its process, then $c(v)=1$,
\item if $(v,v')\in\msg$ and $w'\le v'$ with $P(w')=P(v')$, then there
  exists $(u,u')\in\msg$ with $\lambda(u')=\lambda(v')$, $c(u)=c(u')$,
  and $u'\sim w'$ \cb{(implying $\lambda(u)=\lambda(v)$)},
\item any equivalence class of $\sim$ is finite. 
\end{enumerate}\medskip

\noindent Figure~\ref{fig:color} visualizes the second condition, on
the left, we have the precondition while the right diagram indicates
the conclusion. More precisely, in the precondition, we have a message
$(v,v')\in\msg$ from process $p$ to process $q$ and some node $w'$
preceding $v'$ on the same process. Recall that equivalence classes of
$\sim$ are intervals on process lines. The borders of the equivalence
class containing $w'$ are indicated. Then, by the conclusion, there is
a message $(u,u')\in\msg$ from $p$ to $q$ such that $u'$ belongs to
the indicated equivalence class of $\sim$ (that also contains $w'$)
and the colors of $u$ and $u'$ are the same (which is not indicated).

In general, there can be messages $(u,u')\in\msg$ such that the colors
of $u$ and $u'$ are different, i.e., $c(u)\neq c(u')$. The second
condition ensures that there are ``many'' messages where the send and
the receive event carry the same color. 

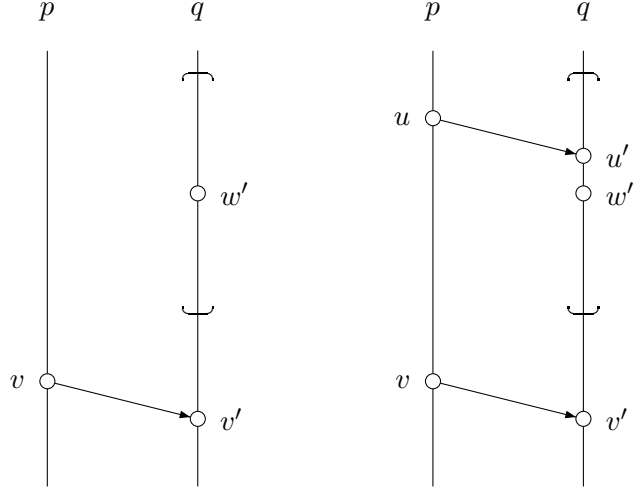
\begin{figure}
  \centering
  \begin{picture}(50,70)
    \gasset{Nframe=y,Nfill=n,Nw=2,Nh=2,Nmr=2,ExtNL=y,NLangle=0,NLdist=2}
    \node[NLangle=180](v)(10,20){$v$} \node(v')(30,15){$v'$}
    \drawedge(v,v'){} \node(w')(30,45){$w'$}
    \put(30,30){\oval(4,2)[b]} \put(30,60){\oval(4,2)[t]}

    \node[Nframe=n,NLangle=90](1)(10,65){$p$}
    \drawedge[AHnb=0](1,v){}
    \node[Nframe=n](1)(10,5){}
    \drawedge[AHnb=0](v,1){}

    \node[Nframe=n,NLangle=90](1)(30,65){$q$}
    \drawedge[AHnb=0](1,w'){}
    \drawedge[AHnb=0](w',v'){}
    \node[Nframe=n](1)(30,5){}
    \drawedge[AHnb=0](v',1){}
  \end{picture}
  \begin{picture}(50,70)
    \gasset{Nframe=y,Nfill=n,Nw=2,Nh=2,Nmr=2,ExtNL=y,NLangle=0,NLdist=2}
    \node[NLangle=180](v)(10,20){$v$} \node(v')(30,15){$v'$}
    \drawedge(v,v'){} \node(w')(30,45){$w'$}
    \put(30,30){\oval(4,2)[b]} \put(30,60){\oval(4,2)[t]}
    \node[NLangle=180](u)(10,55){$u$} \node(u')(30,50){$u'$}
    \drawedge(u,u'){}

    \node[Nframe=n,NLangle=90](1)(10,65){$p$}
    \drawedge[AHnb=0](1,u){}
    \drawedge[AHnb=0](u,v){}
    \node[Nframe=n](1)(10,5){}
    \drawedge[AHnb=0](v,1){}

    \node[Nframe=n,NLangle=90](1)(30,65){$q$}
    \drawedge[AHnb=0](1,u'){}
    \drawedge[AHnb=0](u',w'){}
    \drawedge[AHnb=0](w',v'){}
    \node[Nframe=n](1)(30,5){}
    \drawedge[AHnb=0](v',1){}
  \end{picture}
  \caption{The second condition.} 
  \label{fig:color}
\end{figure}

\begin{prop}\label{P-Col-regular}
  There exists a CFM~$\cA_\Col$ that accepts the set $\Col$.  The
  CFM~$\cA_\Col$ has two messages and its number of local states is
  in~$2^{O(|\cP|)}$. 
\end{prop}

\begin{proof}
  Since the language $\Col$ consists of pairs $(M,c)$, any process~$p$
  \cb{of a CFM with two messages} executes a sequence of events from
  $((\Sigma_p\times\{0,1\})\times\{0,1\})^\infty$ (with
  $\Gamma^\infty$ the set of finite and infinite words over~$\Gamma$)
  where $((\sigma,a),b)$ stands for a $(\sigma,a)$-labeled event that
  sends or receives~$b$. Our automaton $\cA_\Col$ will always send the
  current value of the mapping $c$, i.e., the set of control messages
  is $\{0,1\}$ and we will only execute events from
  \[
  \Gamma_p=\{((p!q,a),a)\mid q\in\cb{\cP\setminus\{p\}},a\in\{0,1\}\}\cup
  \{((p?q,a),b)\mid q\in\cb{\cP\setminus\{p\}},a,b\in\{0,1\}\}. 
  \]
  Having this in mind, consider for $p\in\cP$, $B\subseteq\cP$, and
  $i\in\{0,1\}$ the language $L^p_{B,i}\subseteq\Gamma_p^*$ with $w\in
  L^p_{B,i}$ iff
  \begin{enumerate}[$\bullet$]
  \item if $((p!q,a),a)$ occurs in $w$, then $a=i$,
  \item if $((p?q,a),b)$ occurs in $w$, then $a=i$ and $q\in B$,
  \item for all $q\in B$, the letter $((p?q,i),i)$ occurs in $w$. 
  \end{enumerate}
  Then $L^p_{B,i}$ is regular and can be accepted by a finite
  deterministic automaton $\cB^p_{B,i}$ with $2^{|B|}$ many states. We
  build the $p$-component $\cA_p$ of $\cA_\Col$ from the disjoint
  union of all these automata $\cB^p_{B,i}$ -- it therefore has
  \[
  \sum_{B\subseteq\cP}2\cdot2^{|B|}\le 2\cdot 4^{|\cP|}
  \]
  many states. More precisely, $\cA_p$ is obtained from this disjoint
  union by adding $\varepsilon$-transitions from any accepting state of
  $\cA^p_{B,i}$ to any initial state of $\cA^p_{C,j}$ iff $B\supseteq
  C$ and $i\neq j$. The initial states of $\cA_p$ are the initial
  states of $\cB^p_{B,1}$. A finite run is accepting if it ends in
  some final state of one of the automata $\cB^p_{B,i}$, an infinite
  run is accepting if it takes infinitely many
  $\varepsilon$-transitions. 

  Note that a word $((p\theta_n q_n,a_n),b_n)_{0\le
    n<N}\in\Gamma_p^\infty$ is accepted by $\cA_p$ iff
  \begin{enumerate}[$\bullet$]
  \item $a_0=1$
  \item if $\theta_n=\mathord{!}$, then $a_n=b_n$
  \item if $\theta_n=\mathord{?}$ and $m\le n$, then there exists $k\in\N$ with
    $p\theta_k q_k=p?q_n$, $a_k=b_k$, and, for all $\ell$ in between
    $m$ and $k$, we have $a_m=a_\ell=a_k$. 
  \end{enumerate}
  Hence the CFM consisting of these components accepts the
  language~$\Col$. 
\end{proof}

The \emph{index} $\ind(v)$ of a node \cb{$v\in V_p$} is the maximal number of
mutually non-equivalent nodes from $V_p$ below $v$. Note that
$c(v)=\ind(v)\bmod 2$ for all nodes $v$ if the pair $(M,c)$
satisfies~(1) in the definition of the language~$\Col$. 

\begin{lem}\label{L-msg-and-color}
  Let $(M,c)\in\Col$. Then, for any $(v,v')\in\msg$ with
  $\ind(v)<\ind(v')$, we have $c(v)\neq c(v')$. 
\end{lem}

\begin{proof}
  Suppose there is $(v,v')\in\msg$ with $\ind(v)<\ind(v')$ but
  $c(v)=c(v')$. Since any element of~$M$ dominates a finite set, we
  can assume $v'$ to be minimal with this problem. If
  $\ind(v)+1=\ind(v')$, we are done since $c(v)=\ind(v)\bmod 2\neq
  (\ind(v)+1)\bmod 2= c(v')$. So let $\ind(v)+1<\ind(v')$. Since
  $(M,c)\in\Col$ and $\ind(v')-1>\ind(v)\ge1$, there exists
  $(u,u')\in\msg$ with $\lambda(u')=\lambda(v')$, $c(u)=c(u')$, and
  $\ind(u')=\ind(v')-1$. In particular, $u'<v'$ and therefore $u<v$. 
  But then $\ind(u)\le\ind(v)$. Now we have
  $\ind(u)\le\ind(v)<\ind(v')-1=\ind(u')$, i.e., $u'<v'$ is another
  counterexample to the statement of the lemma. But this contradicts
  the choice of $v'$. 
\end{proof}

\begin{cor}\label{C-inf-many-color-changes}
  Let $(M,c)\in\Col$ and let $(v_1,v_2,\dots)$ be some infinite path
  in~$M$. Then there exist infinitely many $i\in\N$ with $c(v_i)\neq
  c(v_{i+1})$. 
\end{cor}

\begin{proof}
  Since $\ind^{-1}(n)$ is finite for any $n\in\N$, there are
  infinitely many $i\in\N$ with $\ind(v_i)<\ind(v_{i+1})$. If
  $(v_i,v_{i+1})\in\proc$, then $\ind(v_{i+1})=\ind(v_i)+1$ and
  therefore $c(v_i)\neq c(v_{i+1})$. If, in the other case,
  $(v_i,v_{i+1})\in\msg$, then by Lemma~\ref{L-msg-and-color}, we get
  $c(v_i)\neq c(v_{i+1})$. 
\end{proof}

\section{Translation of local formulas}\label{sec:locform}
Let $\alpha$ be a local formula of \DLTL.  We will construct a
``small'' CFM that accepts a pair $(M,c)$ with $M$ an MSC and
$c:V\to\{0,1\}$ iff $c$ is the characteristic function of the set of
positions satisfying $\alpha$, i.e.,
\[
c(v)=
\begin{cases}
  1 & \text{ if }M,v\models\alpha\\
  0 & \text{ otherwise.} 
\end{cases}
\]
To obtain this CFM, we will first construct another CFM that accepts
$(M,(c_\beta)_{\beta\in\sub(\alpha)})$ iff, for all positions $v\in V$
and all subformulas $\beta$ of $\alpha$, we have $M,v\models\beta$ iff
$c_\beta(v)=1$. This CFM will consist of several CFMs running in
conjunction, one for each subformula. For instance, if
$\sigma\in\Sigma$ and $\delta=\beta\lor\gamma$ are subformulas of
$\alpha$, then we will have sub-CFMs that check whether, for any
position $v$, we have $c_\sigma(v)=1$ iff $\lambda(v)=\sigma$ and
$c_\delta(v)=c_\beta(v)\lor c_\gamma(v)$, respectively. We first define these
sub-CFMs for subformulas of the form $\sigma$, $\beta\lor\gamma$, and
$\neg\beta$. 

\begin{exa}\label{E-atomic-automata}
  For $\sigma\in\Sigma$, we define the CFM
  $\cA_\sigma=(\{m\},1,(\cA_p)_{p\in\cP},F)$ as follows: For $p\in\cP$, let
  $S_p=\{\iota_p\}$ and $(\iota_p,\tau,b,m,\iota_p)\in{\to_p}$ iff
  \begin{enumerate}[$\bullet$]
  \item $\tau=\sigma$ and $b=1$ or
  \item $\tau\neq\sigma$ and $b=0$. 
  \end{enumerate}
  Furthermore, $F=\{(\iota_p)_{p\in\cP}\}$. Then it is easily checked
  that $(M,c)$ is accepted by $\cA_\sigma$ iff
  \[
  \forall v\in V: \lambda(v)=\sigma\iff c(v)=1. 
  \]
\end{exa}

\begin{exa}\label{E-BC-of-automata}
  Next we define a CFM $\cA_\lor=(\{m\},3,(\cA_p)_{p\in\cP},F)$: For
  $p\in\cP$, let $S_p=\{\iota_p\}$ and
  $(\iota_p,\tau,(b_1,b_2,b_3),m,\iota_p)\in{\to_p}$ iff $b_3=b_1\lor
  b_2$.  Furthermore, $F=\{(\iota_p)_{p\in\cP}\}$. Then it is easily
  checked that $(M,c)$ is accepted by $\cA_\lor$ iff
  \[
  \forall v\in V: c_3(v)=c_1(v)\lor c_2(v). 
  \]
  The CFM $\cA_\neg$ is defined similarly. 
\end{exa}

\begin{exa}\label{E-E-and-A}
  Next we define a CFM $\cA_\rE=(\{m\},1,(\cA_p)_{p\in\cP},F)$: For $p\in\cP$,
  let $S_p=\{\iota_p,s_p\}$ and $\to_p$ contain precisely
  $(\iota_p,\tau,0,m,\iota_p)$, $(\iota_p,\tau,1,m,s_p)$, and
  $(s_p,\tau,b,m,s_p)$ for all $\tau\in\Sigma_p$ and $b\in\{0,1\}$.
  Furthermore, $F$ is the set of tuples $(f_p)_{p\in\cP}$ that contain at
  least one occurrence of~$s_p$. Hence the index of this CFM is the number of
  processes~$|\cP|$.

  Then it is easily checked that $(M,c)$ is accepted by $\cA_\rE$ iff
  there exists a node $v$ with $c(v)=1$. 

  The CFM $\cA_\rA=(\{m\},1,(\cA_p)_{p\in\cP},F)$ has again just one local
  state per process (and is therefore of index~$1$): For $p\in\cP$, let
  $S_p=\{\iota_p\}$, $\mathord{\to_p} = \{(\iota_p,\tau,1,m,\iota_p)\} \mid
  \tau\in\Sigma_p\}$, and $F=\{(\iota_p)_{p\in\cP}\}$.

  Then it is easily checked that $(M,c)$ admits a run and is therefore
  accepted by $\cA_\rA$ iff $c(v)=1$ for all nodes $v$. 
\end{exa}

\subsection{The backward-path automaton}

Let $\pi$ be a path expression, i.e., a regular expression over the alphabet
$\{\proc,\msg,\{\alpha_1\},\dots,\{\alpha_n\}\}$. Replacing $\{\alpha_i\}$ by
$i$, we obtain a regular expression over the alphabet
$\Gamma=\{\proc,\msg,1,2,\dots,n\}$. Let $L_\pi\subseteq\Gamma^*$ be the
language of this regular expression. 

A word over $\Gamma$ together with a node from an MSC describes a path
starting in that node that walks \emph{backwards}.  The letters
$\proc$ and $\msg$ denote the direction of the path, the letters $i$
denote requirements about the node currently visited (namely, that
$\alpha_i$ shall hold). This idea motivates the following definition:

\begin{defi}\label{def:Lpi-back}
  For an MSC $M$, functions $c_1,\dots,c_n:V\to\{0,1\}$, a node $v\in
  V$ and a word $W\in\Gamma^*$, we define inductively
  $(M,c_1,\dots,c_n),v\models^{-1} W$:
  \begin{align*}
    (M,c_1,\dots,c_n),v&\models^{-1} \varepsilon&\\
    (M,c_1,\dots,c_n),v&\models^{-1} \proc\, W \iff \text{ there
      is }v'=\proc^{-1}(v)
    \text{ with } (M,c_1,\dots,c_n),v'\models^{-1} W\\
    (M,c_1,\dots,c_n),v&\models^{-1} \msg\, W \iff \text{ there is
    }v'=\msg^{-1}(v)
    \text{ with } (M,c_1,\dots,c_n),v'\models^{-1} W\\
    (M,c_1,\dots,c_n),v&\models^{-1} i\, W
    \iff c_i(v)=1\text{ and }(M,c_1,\dots,c_n),v\models^{-1} W\\
  \end{align*}
\end{defi}

We easily verify that $M,v\models\left<\pi\right>^{-1}\true$ iff there exists
$W\in L_\pi$ such that $M,v\models^{-1} W$. 

Let $\Cc=(Q,\iota,T,G)$ be a finite automaton over~$\Gamma$
recognizing~$L_\pi$. Note that we can assume $|Q|\in O(s(\pi))$. For $q\in Q$
and $W\in\Gamma^*$, we write $q.W\subseteq Q$ for the set of states that can
be reached from $q$ reading the word~$W$, and we denote by $W.q \subseteq Q$
the set of states from which one can reach $q$ when reading $W$. Furthermore,
$P.L=\bigcup_{p\in P, W\in L}p.W$ and $L.P=\bigcup_{W\in L, p\in P}W.p$ for
$P\subseteq Q$ and $L\subseteq\Gamma^*$ (if $L$ (or $P$) is a singleton, then
we may identify it with its unique element). 

\begin{lem}\label{L-pi-back-alpha}
  There exists a CFM $\cA$ with sets of local states $2^Q$ and set of messages
  $2^Q$ such that, for any run $\rho$ of $\cA$ on $(M,c_1,\dots,c_n)$ and any
  node $v$ of $M$, we have \linebreak$\rho(v)=\{q\in Q\mid\exists
  W\in\Gamma^*:q \in W.G\text{ and }M,v\models^{-1}W\}$. 
\end{lem}

\begin{proof}
  To define the set of transitions, let $A_1,A_2,A_2'\subseteq 2^Q$, and let
  $a=(\sigma,b_1,\dots,b_n)\in\Sigma_p\times\{0,1\}^{n}$ and
  $N=\{i\in[n]\mid b_i=1\}$. Then we set
  \[
  A_1\xrightarrow{a,A'_2}_p A_2
  \]
  iff the following conditions hold
  \begin{enumerate}[(1)]
  \item if $\sigma$ is a send action, then $A_2=A_2'=N^*.G\cup N^*\,\proc.A_1$,
  \item if $\sigma$ is a receive action, then $A_2=N^*.G\cup N^*\,\proc.A_1\cup
    N^\ast\,\msg.A_2'$. 
  \end{enumerate}
  Here, $A_1$ is the local state assumed before the execution of the
  (labeled) action~$a$, $A_2$ is the local state assumed afterwards,
  and $A_2'$ is the message involved in this transition. Depending on
  whether $a$ is a receive or a send action, the message is consumed
  by $a$ or emitted by $a$. These two situations are visualized in
  Figure~\ref{fig:transitions}. 

  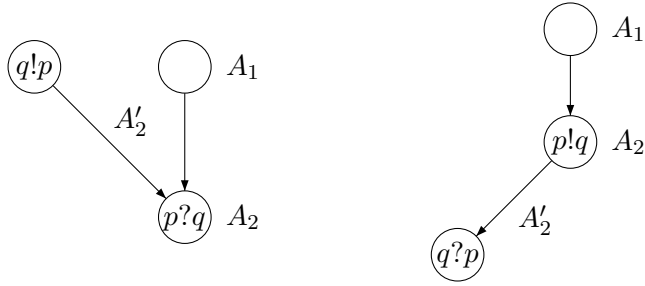
\begin{figure}
    \centering
    \begin{picture}(50,40)
      \gasset{Nframe=y,Nfill=n,Nw=7,Nh=7,Nmr=7,ExtNL=y,NLdist=2}
      \node[NLangle=0](1)(30,30){$A_1$}
      \node[NLangle=0](2)(30,10){$A_2$}
      \nodelabel[ExtNL=n,NLdist=0](2){$p?q$}
      \node[ExtNL=n,NLdist=0](3)(10,30){$q!p$} \drawedge(1,2){}
      \drawedge(3,2){$A_2'$}
    \end{picture}
    \begin{picture}(50,40)
      \gasset{Nframe=y,Nfill=n,Nw=7,Nh=7,Nmr=7,ExtNL=y,NLdist=2}
      \node[NLangle=0](1)(30,35){$A_1$}
      \node[NLangle=0](2)(30,20){$A_2$}
      \nodelabel[ExtNL=n,NLdist=0](2){$p!q$}
      \node[ExtNL=n,NLdist=0](3)(15,5){$q?p$} \drawedge(1,2){}
      \drawedge(2,3){$A_2'$}
    \end{picture}
    \caption{Transitions of the CFM.} 
    \label{fig:transitions}
  \end{figure}

  The local initial state is \cb{$\emptyset$} for any $p\in\cP$ and
  any tuple of local states is accepting. Now let $(\rho,\mu)$ be a
  run of this CFM on $(M,c_1,\dots,c_n)$. 

  For $v\in V$ set $N_v=\{i\in[n]\mid c_i(v)=1\}$ and
  $M(v)={\{W\in\Gamma^*\mid (M,c_1,\dots,c_n),v\models^{-1}W\}}$. Then
  it is easily verified that
  \[
  M(v)=N_v^*\cup \underbrace{N_v^*\,\proc\,M(\proc^{-1}(v))}
  _{\text{if $\proc^{-1}(v)$ is defined}} \cup
  \underbrace{N_v^*\,\msg\,M(\msg^{-1}(v))} _{\text{if $\msg^{-1}(v)$
      is defined}}\ . 
  \]
  On the other hand,
  \[
  \rho(v)= N_v^*.G \cup
  \underbrace{N_v^*.\proc.\rho(\proc^{-1}(v))} _{\text{if
      $\proc^{-1}(v)$ is defined}} \cup
  \underbrace{N_v^*.\msg.\rho(\msg^{-1}(v))}_{\text{if $\msg^{-1}(v)$
      is defined}}\ . 
  \]
  Hence, by induction on the partial order $(V,\le)$, we have $q\in\rho(v)$
  iff $q\in M(v).G$. 
\end{proof}

\begin{thm}\label{T-pi-back-alpha}
  Let $\left<\pi\right>^{-1}\alpha$ be a local formula such that $\pi$
  is a regular expression over the alphabet
  $\{\proc,\msg,\{\alpha_1\},\dots,\{\alpha_n\}\}$. Then there exists
  a CFM $\cA_{\left<\pi\right>^{-1}\alpha}$ of index $1$ with the following
  property: Let $M$ be an MSC and let $c_i:V\to\{0,1\}$ be the
  characteristic function of the set of positions satisfying
  $\alpha_i$ (for all $i\in[n+1]$) where $\alpha_{n+1}=\alpha$. Then
  $(M,c_1,\dots,c_n,c_{n+1},c)$ is accepted iff $c$ is the
  characteristic function of the set of positions satisfying
  $\left<\pi\right>^{-1}\alpha$. 

  The CFM we construct has $2^{O(s(\pi))}$ local states per process,
  $2^{O(s(\pi))}$ many control messages, and any tuple of local states
  is accepting (in particular, the CFM has index~$1$). 
\end{thm}

\begin{proof}
  Again, since $M,v\models\left<\pi\right>^{-1}\alpha$ iff
  $M,v\models\left<\pi;\{\alpha\}\right>^{-1}\true$, we will
  assume~$\alpha=\true$. The CFM $\cA_{\left<\pi\right>^{-1}\alpha}$ simulates
  the run $(\rho,\mu)$ of the CFM~$\cA$ from Lemma~\ref{L-pi-back-alpha} and
  verifies that $c(v)=1$ iff $\iota \in \rho(v)$ for all nodes $v\in V$. Then
  we have
  \begin{align*}
    c(v)=1 &\iff \iota \in \rho(v)\\
    &\iff \exists W\in\Gamma^*: \iota \in W.G
    \text{ and }M,v\models^{-1}W\\
    &\iff \exists W\in L(\Cc)=L_\pi: M,v\models^{-1}W\\
    &\iff \exists W\in L_\pi: M,v\models^{-1}W\\
    &\iff M,v\models\left<\pi\right>^{-1}\true 
  \end{align*}
  This concludes the proof of Theorem~\ref{T-pi-back-alpha}. 
\end{proof}

\subsection{The forward-path automaton}
\label{ssec:forward}

We now turn to a similar CFM corresponding to subformulas of the
form~$\left<\pi\right>\true$. We will prove the following analog to
Theorem~\ref{T-pi-back-alpha}. This proof will, however, be substantially more
difficult. 


\begin{thm}\label{T-pi-alpha}
  Let $I\subseteq\Ch$ and let $\left<\pi\right>\alpha$ be a local
  formula such that $\pi$ is a regular expression over the alphabet
  $\{\proc,\msg,\{\alpha_1\},\dots,\{\alpha_n\}\}$. Then there exists
  a CFM $\cA_{\left<\pi\right>\alpha}$ of index~$1$ with the following
  property: Let $M$ be an MSC with $\Inf(M)=I$ and let
  $c_i:V\to\{0,1\}$ be the characteristic function of the set of
  positions satisfying $\alpha_i$ (for all $i\in[n+1]$) where
  $\alpha_{n+1}=\alpha$. Then $(M,c_1,\dots,c_n,c_{n+1},c)$ is
  accepted iff $c$ is the characteristic function of the set of
  positions satisfying $\left<\pi\right>\alpha$. The CFM we construct
  has $2^{O(s(\pi)+|\cP|)}$ local states per process and
  $2^{O(s(\pi))}$ many control messages. 
\end{thm}

The rest of this section is devoted to the proof of this theorem. 
Since $M,v\models\left<\pi\right>\alpha$ iff
$M,v\models\left<\pi;\{\alpha\}\right>\true$, we will
assume~$\alpha=\true$, i.e., $\alpha$ holds true for any node of any
MSC. 

Let $\Gamma=\{\proc,\msg,1,2,\dots,n\}$. If, in the regular
expression~$\pi$, we replace any occurrence of $\{\alpha_i\}$ by $i$,
we obtain a regular expression over the alphabet $\Gamma$. Let
$L_\pi\subseteq\Gamma^*$ be the language denoted by this regular
expression. Then there is a finite automaton $\Cc=(Q,\iota,T,G)$
over~$\Gamma$ with set of states~$Q$, initial state~$\iota$, set of
transitions~$T$, and set of final states~$G$ \cb{recognizing $L_\pi$}. Note that $|Q|\in
O(s(\pi))$. 

A word over $\Gamma$ together with a node from an MSC describe a path
starting in that node walking \emph{forwards}. The following is
therefore the forward-version of Def.~\ref{def:Lpi-back}. 

\begin{defi}\label{def:Lpi}
  For an MSC $M$, functions $c_1,\dots,c_n:V\to\{0,1\}$, a node $v\in
  V$ and a word $W\in\Gamma^*$, we define inductively
  $(M,c_1,\dots,c_n),v\models W$:
  \begin{align*}
    (M,c_1,\dots,c_n),v&\models \varepsilon&\\
    (M,c_1,\dots,c_n),v&\models \proc\, W \iff \text{ there exists
    }v'=\proc(v)
    \text{ with } (M,c_1,\dots,c_n),v'\models W\\
    (M,c_1,\dots,c_n),v&\models \msg\, W \iff \text{ there exists
    }v'=\msg(v)
    \text{ with } (M,c_1,\dots,c_n),v'\models W\\
    (M,c_1,\dots,c_n),v&\models i\, W
    \iff c_i(v)=1\text{ and }(M,c_1,\dots,c_n),v\models W\\
  \end{align*}
\end{defi}

Now the following is immediate. 
\begin{lem}
  Let $M$ be an MSC and, for $i\in[n]$, let $c_i:V\to\{0,1\}$ be the
  characteristic function of the set of positions satisfying
  $\alpha_i$. Then $M,v\models\left<\pi\right>\true$ iff there exists
  $W\in L_\pi$ such that $M,v\models W$. 
\end{lem}

Thus, in order to prove Theorem~\ref{T-pi-alpha}, it suffices to
construct a CFM that accepts $(M,c_1,\dots,c_n,c)$ iff
\begin{align*}
  &\forall v\in V:c(v)=0~\Longrightarrow~\forall
  W\in L_\pi:(M,c_1,\dots,c_n),v\not\models W\\
  \land &\forall v\in V:c(v)=1~\Longrightarrow~\exists W\in
  L_\pi:(M,c_1,\dots,c_n),v\models W. 
\end{align*}
Since the class of languages accepted by CFMs is closed under
intersection, we can handle the two implications separately in the
following two subsections. 

\subsubsection{Any $0$ is justified}

We construct a CFM that accepts $(M,c_1,\dots,c_n,c)$ iff, for any
$v\in V$ with $c(v)=0$, there does not exist $W\in L_\pi$ with
$(M,c_1,\dots,c_n,c),v\models W$. The basic idea is rather simple:
whenever the CFM encounters a node $v$ with $c(v)=0$, it will start
the automaton $\Cc$ (that accepts $L_\pi$) and check that it cannot
reach an accepting state whatever path we choose starting in~$v$. 
Since the CFM has to verify more than one $0$, the set of local
states~$S_p$ equals $2^{Q\setminus G}$ with initial state
$\iota_p=\emptyset$ for any $p\in\cP$. The set of control messages~$C$
equals~$2^{Q\setminus G}$, too. Furthermore, any tuple of local states
is accepting. 

To define the set of transitions, let $A_1,A_2\in S_p$ and $A_2'\in C$. 
Moreover, let $a=(\sigma,b_1,\dots,b_n,b)\in\Sigma_p\times\{0,1\}^{n+1}$ and
$N=\{i\in[n]\mid b_i=1\}$. Then we have a transition
\[
A_1\xrightarrow{a,A'_2}_p A_2
\]
iff the following conditions hold:
\begin{enumerate}[(1)]
\item if $b=0$, then $\iota.N^*\subseteq A_2$,
\item $A_1.\proc.N^*\subseteq A_2$,
\item if $\sigma$ is a receive action, then $A'_2.\msg.N^*\subseteq
  A_2$,
\item if $\sigma$ is a send action, then $A'_2=A_2$. 
\end{enumerate}

\begin{lem}
  Let $(\rho,\mu)$ be a run of the above CFM on $(M,c_1,\dots,c_n,c)$
  and let $v_0\in V$ with $c(v_0)=0$. Then there does not exist $W\in
  L_\pi$ with $(M,c_1,\dots,c_n),v_0\models W$. 
\end{lem}

\begin{proof}
  Suppose there is $W\in L_\pi$ with $(M,c_1,\dots,c_n),v_0\models W$. 
  Write $W=w_0a_1w_1\dots a_nw_m$ with $a_k\in\{\proc,\msg\}$ and
  $w_k\in[n]^*$ for all appropriate~$k$. Since
  $(M,c_1,\dots,c_n),v_0\models W$, there exist nodes $v_k\in V$ with
  $v_{k+1}=a_{k+1}(v_k)$ and $w_k\subseteq N_{v_k}^*$ where
  $N_{v_k}=\{i\in[n]\mid c_i(v_k)=1\}$. Since $W\in L_\pi$, there are
  states $q_i\in Q$ with $q_0\in\iota.w_0$, $q_{i+1}\in
  q_i.a_{i+1}w_{i+1}$, and $q_m\in G$. 

  Since $c(v_0)=0$, we have $\iota.N_{v_0}^*\subseteq\rho(v_0)$ by~(1)
  and therefore $q_0\in\iota.w_0\subseteq\rho(v_0)$ by $w_0\in
  N_{v_0}^*$. By induction, assume $k<m$ and $q_k\in\rho(v_k)$. If
  $a_{k+1}=\proc$, then by (2) $q_{k+1}\in
  q_k.\proc.N_{v_{k+1}}^\ast\subseteq \rho(v_{k+1})$. If
  $a_{k+1}=\msg$, then $v_k$ is a send event. Hence, by~(4),
  $\mu(v_k)=\rho(v_k)$. Since $(v_k,v_{k+1})\in\msg$, this implies
  $\mu(v_{k+1})=\rho(v_k)$. Hence, by~(3),
  $q_{k+1}\in\rho(v_k).\msg.N_{v_{k+1}}^*\subseteq\rho(v_{k+1})$. This
  finishes the inductive argument. Hence $q_m\in\rho(v_n)\cap G$,
  contradicting our definition $S_p=2^{Q\setminus G}$. 
\end{proof}

\begin{lem}
  Suppose $(M,c_1,\dots,c_n,c)$ satisfies
  \[
  \forall v\in V:c(v)=0 ~\Longrightarrow~\forall W\in
  L_\pi:(M,c_1,\dots,c_n),v\not\models W. 
  \]
  Then $(M,c_1,\dots,c_n,c)$ admits a run of the above CFM. 
\end{lem}

\begin{proof}
  For $v\in V$, let $N_v=\{i\in[n]\mid c_i(v)=1\}$. Then define
  $\rho(v)$ to be the union of the following sets
  \begin{enumerate}[(a)]
  \item $\iota.N_v^*$ if $c(v)=0$,
  \item $\rho(\proc^{-1}(v)).\proc.N_v^*$ if $\proc^{-1}(v)$ is
    defined (i.e., if $v$ is not minimal on its process),
  \item $\rho(\msg^{-1}(v)).\msg.N_v^*$ if $\msg^{-1}(v)$ is defined
    (i.e., if $\lambda(v)$ is a receive action). 
  \end{enumerate}
  Furthermore, let
  \[
  \mu(v)=
  \begin{cases}
    \rho(v) & \text{ if $\lambda(v)$ is a send action}\\
    \rho(\msg^{-1}(v)) & \text{ otherwise.} 
  \end{cases}
  \]
  Then, for any $v\in V$, the transition conditions (1-4) are
  satisfied by the mappings $\rho$ and $\mu$ (recall that the local
  initial states are $\emptyset$). 

  Now, by contradiction, assume $(\rho,\mu)$ is no run, i.e., there is
  some $v_0\in V$ with $\rho(v_0)\notin 2^{Q\setminus G}$. Hence there
  exists $q_0\in\rho(v_0)\cap G$. Setting $W_0=\varepsilon$, we
  therefore have \\
  \hspace*{\fill}
  \begin{minipage}{.7\linewidth}
    $(M,c_1,\dots,c_n,c),v_k\models W_k$, $q_k\in\rho(v_k)$, and
    $q_k.W_k\cap G\neq\emptyset$
  \end{minipage}
  \hspace{\fill}(*)\\
  for $k=0$. Now assume that (*) holds for some~$k\ge0$. 

  First, assume $c(v_k)=0$ and $q_k\in\iota.N_{v_k}^*\subseteq \rho(v_k)$
  because of (a). Hence, there exists $w_k\in N_{v_k}^*$ with
  $q_k\in\iota.w_k$. But then
  $(M,c_1,\dots,c_n,c),v_k\models w_kW_k$ and $w_kW_k\in L_\pi$, a
  contradiction. Hence we have $q_k\in\rho(v_k)$ because of~(b)
  or~(c). If $q_k\in\rho(\proc^{-1}(v_k)).\proc.N_{v_k}^*$, then set
  $v_{k+1}=\proc^{-1}(v_k)$ and choose $q_{k+1}\in\rho(v_{k+1})$ and
  $w_k\in N_{v_k}^*$ with $q_k\in q_{k+1}.\proc.w_k$. Setting
  $W_{k+1}=\proc.w_k.W_k$ yields (*) for $k+1$. If
  $q_k\in\rho(\msg^{-1}(v_k)).\msg.N_{v_k}^*$, we can argue similarly. 

  Hence we find an infinite sequence of nodes $v_0>v_1>v_2\dots$ which
  is impossible since $v_0$ dominates only a finite set. Thus,
  $(\rho,\mu)$ is a run. 
\end{proof}

\begin{prop}\label{P-any-0-justified}
  There exists a CFM~$\cA_0$ of index~$1$ that accepts
  $(M,c_1,\dots,c_n,c)$ iff
  \[
  \forall v\in V:c(v)=0~\Longrightarrow~\forall W\in L_\pi:
  (M,c_1,\dots,c_n),v\not\models W. 
  \]
  The number of local states per process as well as the number of
  messages are in~$2^{O(s(\pi))}$. Furthermore, any run of the CFM is
  accepting. 
\end{prop}

\begin{proof}
  The proof is immediate by the above two lemmas. 
\end{proof}

\subsubsection{Any $1$ is justified}

We next construct a CFM that accepts $(M,c_1,\dots,c_n,c)$ iff, for
any $v\in V$ with $c(v)=1$, there exists $W\in L_\pi$ with
$(M,c_1,\dots,c_n,c),v\models W$. Again, the basic idea is simple:
whenever the CFM encounters a node $v$ with $c(v)=1$, it will start
the automaton $\Cc$ (that accepts~$L_\pi$) and check that it can reach
an accepting state along one of the possible paths.  Thus, before, we
had to prevent $\Cc$ from reaching an accepting state. This time, we
have to ensure that any verification of a $c(v)=1$ will eventually
result in an accepting state being reached. For sequential
B\"uchi-automata, solutions to this problem are known: collect some
claims to be verified in one set and, only when all of them are
verified, start verifying those claims that have been encountered
during the previous verification phase. The resulting
B\"uchi-automaton accepts iff the verification phase is changed
infinitely often. We will adapt precisely this idea here. But then,
the CFM would have to accept if, along \emph{each and every} path, the
verification phase changes infinitely often.  This is the point where
the CFM $\cA_\Col$ comes into play since, by
Corollary~\ref{C-inf-many-color-changes}, it verifies that any path
runs through infinitely many color changes. Thus, we will first
construct a CFM that runs on tuples $(M,c_0,c_1,\dots,c_n,c)$ where we
assume that $(M,c_0)\in\Col$. The actual CFM that verifies all claims
$c(v)=1$ will run this newly constructed CFM in conjunction with
$\cA_\Col$ (that verifies $(M,c_0)\in\Col$) and project away the
labeling~$c_0$. 

For any $p\in\cP$, the set of local states~$S_p$ equals $2^Q\times
2^Q\times\{0,1\}$ with initial state
$\iota_p=(\emptyset,\emptyset,1)$, the set of control messages~$C$
equals $2^Q\times 2^Q\times\{0,1\}$. 

To define the set of transitions, let $(A_1,B_1,d_1), (A_2,B_2,d_2)\in
S_p$ and $(A'_2,B'_2,d_2')\in C$. Furthermore, let
$a=(\sigma,b_0,b_1,\dots,b_n,b)\in\Sigma_p\times\{0,1\}^{n+2}$. Now we would like to define the conditions for the existence of a transition
\[
(A_1,B_1,d_1)\xrightarrow{a,(A'_2,B'_2,d_2')}_p(A_2,B_2,d_2)\,. 
\]
We have to distinguish between $\sigma$ being a send or a receive
event, cf.\ Figure~\ref{fig:transitions_forward}. For $\sigma= p!q$
the pair $(A_1,B_1)$ contains the in-going information whereas
$(A_2,B_2)$ and $(A_2',B_2')$ carry the out-going information
propagated along the process and the channel, respectively. On the
other hand, for $\sigma=p?q$ now both $(A_1,B_1)$ and  $(A_2',B_2')$
contain the in-going information whereas the out-going information can
be propagated along the process line only, hence, it is enclosed in $(A_2,B_2)$ only. Therefore, we put
\begin{align*}
  A_\rein = A_1,\; A_\raus= A_2\cup A_2',\;  B_\rein = B_1,\;  B_\raus= B_2\cup B_2'
\end{align*}
whenever $\sigma$ is a send event and
\begin{align*}
 A_\rein = A_1\cup A_2',\;  A_\raus= A_2, \; B_\rein = B_1\cup B_2',\;  B_\raus= B_2
\end{align*}
whenever $\sigma$ is a receive event. 
 \begin{figure}
    \centering
    \begin{picture}(60,45)(0,-5)
      \gasset{Nframe=y,Nfill=n,Nw=7,Nh=7,Nmr=7,ExtNL=y,NLdist=2}
      \node[NLangle=0](1)(30,35){$(A_1,B_1,d_1)$}
      \node[NLangle=0](2)(30,20){$(A_2,B_2,d_2)$}
      \nodelabel[ExtNL=n,NLdist=0](2){$p!q$}
      \node[ExtNL=n,NLdist=0](3)(15,5){$q?p$} \drawedge(1,2){}
      \drawedge[ELpos=70,ELside=r](2,3){$(A_2',B_2',d_2')$}
      \node[Nw=.5,Nh=.5,Nframe=n](4)(30,0){}\drawedge(2,4){}
    \end{picture}
    \begin{picture}(60,45)(0,-5)
      \gasset{Nframe=y,Nfill=n,Nw=7,Nh=7,Nmr=7,ExtNL=y,NLdist=2}
      \node[NLangle=0](1)(30,35){$(A_1,B_1,d_1)$}
      \node[NLangle=0](2)(30,15){$(A_2,B_2,d_2)$}
      \nodelabel[ExtNL=n,NLdist=0](2){$p?q$}
      \node[ExtNL=n,NLdist=0](3)(10,35){$q!p$} \drawedge(1,2){}
      \drawedge[ELpos=30,ELside=r](3,2){$(A_2',B_2',d_2')$}
      \node[Nw=.5,Nh=.5,Nframe=n](4)(30,0){}\drawedge(2,4){}
    \end{picture}
    \caption{Transitions of the CFM.} 
    \label{fig:transitions_forward}
  \end{figure}
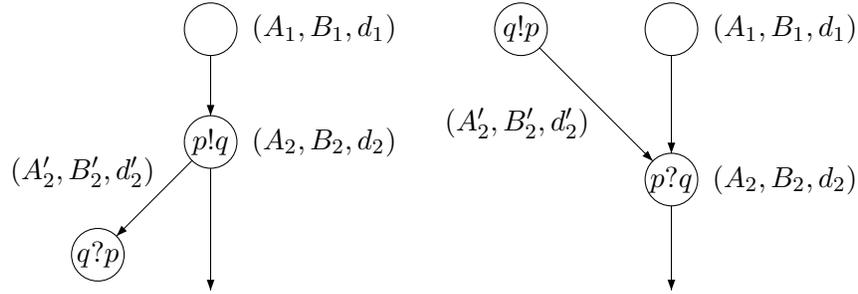

Now the idea is the following: The CFM saves the actual color within
its state and propagates it via the channel whenever $\sigma$ is a
send. Whenever we stay within the same color ($d_2=d_1$ and, for
$\sigma$ a receive, also $d_2=d_2'$) we propagate the states (from
the finite automaton $\Cc$) contained in $A_\rein$ to $A_\raus$ and
likewise from $B_\rein$ to $B_\raus$. But whenever the color changes
($d_2\neq d_1$ or, for $\sigma$ a receive, $d_2\neq d_2'$), we require
the respective part of $A_\rein$ to be empty and all the information
from the respective $B_\rein$ is swept to $A_\raus$. Moreover,
whenever a new $1$ has to be verified we start $\Cc$ and collect the
states obtained this way within $B_\raus$. Now we formalize these
ideas. 

Let $N= \{i\in[n]\mid b_i=1\}$. Recall that $\Cc=(Q,\iota,T,G)$ is the
finite automaton recognizing $L_\pi$. Then the transition above is
defined iff the following conditions hold:
\begin{enumerate}[(1)]
\item $d_2=b_0$ and, if $\sigma$ is a send, then also $d_2'=b_0$, 
\item if $b=1$, then $\iota.N^*\cap (G\cup B_\raus)\neq\emptyset$,
\item $\forall q\in A_1: q.\proc\,N^* \cap (G\cup A_\raus)\neq\emptyset$, and,\\  if $\sigma$ is a receive, then also $\forall q\in A_2': q.\msg\,N^* \cap (G\cup A_\raus)\neq\emptyset$,
\item if $d_2=d_1$, then $\forall q\in B_1: q.\proc\,N^* \cap (G\cup B_\raus)\neq\emptyset$,
\item if $d_2\neq d_1$, then $A_1=\emptyset$ and $\forall q\in B_1: q.\proc\,N^*\cap (G\cup A_\raus)\neq\emptyset$, 
\item if $d_2=d_2'$ and $\sigma$ is a receive, then $\forall q\in B_2': q.\msg\,N^* \cap (G\cup B_\raus)\neq\emptyset$,
\item if $d_2\neq d_2'$ and $\sigma$ is a receive, then $A_2'=\emptyset$ and $\forall q\in B_2': q.\msg\,N^*\cap (G\cup A_\raus)\neq\emptyset$. 
\end{enumerate}
Note that for a color change (cases (5) and (7)) the respective conditions in (3) become obsolete since $A_1=\emptyset$ and/or $A_2'=\emptyset$. 

Recall that $I$ is a set of channels and that we are only interested
in MSCs that use precisely these channels infinitely often. Let
$(f_p)_{p\in\cP}\in\prod_{p\in\cP}S_p$ be accepting in $\cA$ iff
$f_p\in\{(\emptyset,\emptyset,0),(\emptyset,\emptyset,1)\}$ for
all~$p\in \cP$ that are not involved in any of the channels from~$I$,
i.e., that satisfy $I\cap(\{p\}\times\cP\cup\cP\times\{p\})=\emptyset$
(note that a process $p$ is not involved in any of the channels
from~$I$ iff it is not involved in any of the channels used infinitely
often iff $p$ executes only finitely many events). This finishes the
construction of the CFM~$\cA$ of index~$1$. 

\begin{lem}
  Let $(\rho,\mu)$ be an accepting run of the above CFM $\cA$ on
  $(M,c_0,c_1,\dots,c_n,c)$ and suppose $(M,c_0)\in\Col$ and
  $I\subseteq\Inf(M)$. Then, for any $v_0\in V$ with $c(v_0)=1$, there
  exists $W\in L_\pi$ with $(M,c_1,\dots,c_n),v_0\models W$. 
\end{lem}

\begin{proof}
  For $v\in V$, let $\rho(v) = (A_v,B_v,d_v)$, $\mu(v) = (A'_v,B'_v,d_v')$,
  $N_v =\{i\in[n]\mid c_i(v)=1\}$. Similarly as above,  whenever $\lambda(v)$
  is a send event, we put
\begin{align*}
  A_\raus(v)= A_v\cup A_v',\;   B_\raus(v)= B_v\cup B_v',
\end{align*}
and whenever $\lambda(v)$ is a receive event, we put
\begin{align*}
 A_\raus(v)= A_v, \; B_\raus(v)= B_v\,. 
\end{align*}

Since $c(v_0)=1$, (2) implies the existence of $w_0\in N_{v_0}^*$
  and $q_0\in \iota.w_0\cap(G\cup B_\raus)$. Now we define a finite or infinite sequence $(v_i,w_i,q_i)_{0\le i<N}$
  with $N\in\N\cup\{\omega\}$, $v_i\in V$, $w_i\in\{\proc,\msg\}N_{v_i}^*$ for $i\ge 1$, and $q_i\in Q$ such that the following hold for all
  $0\le i<N$:
  \begin{enumerate}[(a)]
  \item $(v_i,v_{i+1})\in\proc\cup\msg$,
  \item $q_i\in G\cup A_\raus(v_i)\cup B_\raus(v_i)$, 
  \item $w_{i+1}\in a\,N_{v_{i+1}}^*$ and $q_{i+1}\in q_i.w_{i+1}$ with
  $a=
    \begin{cases}
      \proc & \text{ if }(v_i,v_{i+1})\in\proc,\\
      \msg & \text{ otherwise,}
    \end{cases}
    $
  \item if $q_i\in A_\raus(v_i)$, then $q_{i+1}\in
    G\cup A_\raus(v_{i+1})$,
  \item $q_i\in G\iff N=i+1$\,. 
  \end{enumerate}

Let us  assume that the sequence has already been constructed up to index $i$. Then we proceed as follows:
  \begin{enumerate}[(i)]
  \item If $q_i\in G$, then set $N=i+1$ which finishes the
    construction of the sequence and implies (e) \textit{a posteriori} for all
    $0\le i <N$. 
  \item Suppose $q_i\in A_{v_i}\setminus G$. Since
    $A_{v_i}\neq\emptyset$ and since the run is accepting, the node
    $v_i$ is not maximal on its process, so we can set
    $v_{i+1}=\proc(v_i)$. Then, by (3), we can choose
    $w_{i+1}\in\proc\,N_{v_{i+1}}^*$ and $q_{i+1}\in
    q_i.w_{i+1}\cap(G\cup A_\raus(v_{i+1}))$. 
  \item Suppose $q_i\in A_\raus(v_i)\setminus(G\cup A_{v_i})$. Then
    $v_i$ is a send event. Hence, $v_{i+1}=\msg(v_i)$ is a
    well-defined receive event. Hence, by~(3), there exist
    $w_{i+1}\in\msg\,N_{v_{i+1}}^*$ and $q_{i+1}\in
    q_i.w_{i+1}\cap(G\cup A_\raus(v_{i+1}))$ such that (a) -- (e)
    hold. 
  \item Suppose $q_i\in B_{v_i}\setminus(G\cup A_\raus(v_i))$. 
    Since $B_{v_i}\neq\emptyset$ and since the run is accepting, the
    node $v_i$ cannot be maximal on its process, i.e.,
    $v_{i+1}=\proc(v_i)$ is well-defined. Then
    \begin{enumerate}[(a)]
    \item if $d_{v_{i+1}}=d_{v_i}$, by (4), $q_i.\proc\,N_{v_{i+1}}^*\cap(G\cup B_\raus(v_{i+1}))\neq\emptyset$, and
  \item if $d_{v_{i+1}}\neq d_{v_i}$, by (5), $q_i.\proc\,N_{v_{i+1}}^*\cap(G\cup A_\raus(v_{i+1}))\neq\emptyset$. 
    \end{enumerate}
    Hence we can choose $w_{i+1}\in \proc\,N_{v_{i+1}}^*$ and $q_{i+1}$ such
    that (a) -- (e) hold. 
  \item Finally, suppose $q_i\in B_\raus(v_i)\setminus(G\cup
    A_\raus(v_i)\cup B_{v_i})$ such that $v_i$ is a send event, i.e.,
    $v_{i+1}=\msg(v_i)$ is a well-defined receive event. Hence
    \begin{enumerate}[(a)]
    \item if $d_{v_{i+1}}=d_{v_{i+1}}'$, by (6), $q_i.\msg\,N_{v_{i+1}}^*\cap(G\cup B_\raus(v_{i+1}))\neq\emptyset$, and
  \item if $d_{v_{i+1}}\neq d_{v_{i+1}}'$, by (7), $q_i.\msg\,N_{v_{i+1}}^*\cap(G\cup A_\raus(v_{i+1}))\neq\emptyset$. 
    \end{enumerate}
Again, we can choose
    $w_{i+1}$ and $q_{i+1}$ such that
    (a) -- (e) hold. 
  \end{enumerate}
  If the construction can be carried out \textit{ad infinitum}, then
  set $N=\omega$ which, again, ensures (e) \textit{a posteriori} for
  all $i<N$. 

  Now suppose $N=\omega$. Then, by
  Cor.~\ref{C-inf-many-color-changes}, there exist $0\le i<k$ with
  $c_0(v_i)\neq c_0(v_{i+1})$ and $c_0(v_k)\neq c_0(v_{k+1})$. By (1),
  this implies $d_{v_i}\neq d_{v_{i+1}}$ whenever $(v_i,v_{i+1})\in\proc$ or
  $d_{v_{i+1}}'\neq d_{v_{i+1}}$ whenever $(v_i,v_{i+1})\in\msg$. Similarly, $d_{v_k}\neq
  d_{v_{k+1}}$ for $(v_k,v_{k+1})\in\proc$ and $d_{v_{k+1}}'\neq d_{v_{k+1}}$
  for $(v_k,v_{k+1})\in\msg$. Let us assume $(v_i,v_{i+1})\in\proc$ and
  $(v_k,v_{k+1})\in\proc$, i.e., $d_{v_i}\neq d_{v_{i+1}}$ and $d_{v_k}\neq
  d_{v_{k+1}}$. Hence, by $(5)$, we get
  $A_{v_i}=A_{v_k}=\emptyset$. Since, by (e), $q_i\notin
  G$, $A_{v_i}=\emptyset$, and $(v_i,v_{i+1})\in \proc$, we get $q_i\in
  B_{v_i}$ and, therefore, by case~(iv)(b) of the construction above
  $q_{i+1}\in A_\raus(v_{i+1})$.  Applying (e) and (d) inductively, this
  results in $q_k\in A_\raus(v_k)$.  Since $(v_k,v_{k+1})\in \proc$, we
  conclude by the construction (cases (ii) and (iii)) $q_k\in A_{v_k}\setminus
  G$. But this contradicts $A_{v_k}=\emptyset$.  For the
  other cases a contradiction is obtained  similarly. Hence, $N$ is finite. 

  Certainly, $(M,c_1,\dots,c_n),v_{N-1}\models\varepsilon$. From (c), we
  obtain $(M,c_1,\dots,c_n),v_{N-2}\models w_{N-1}$ and, by induction,
  $(M,c_1,\dots,c_n),v_0\models W$ with $W=w_0w_1\dots w_{N-1}$. Since
  $q_0\in \iota.w_0$, (c) implies $q_{N-1}\in\iota.W$ and therefore
  $W\in L_\pi$ follows from (e).  
\end{proof}

\begin{lem}\label{L-hhh}
  Suppose $(M,c_1,\dots,c_n,c)$ satisfies
  \[
  \forall v\in V: c(v)=1\Longrightarrow\exists W\in
  L_\pi:(M,c_1,\dots,c_n),v\models W
  \]
  and $\Inf(M)\subseteq I$.  Then there exists a mapping
  $c_0:V\to\{0,1\}$ such that $(M,c_0)\in\Col$ and the above CFM~$\cA$
  accepts $(M,c_0,c_1,\dots,c_n,c)$. 
\end{lem}

\begin{proof}
  For any $v\in V$ with $c(v)=1$, there exist $0\le k^v\in\N$,
  $w_0^v\in[n]^*$, $w_i^v\in\{\proc,\msg\}[n]^*$ for $1\le i\le k^v$,
  and $q_i^v\in Q$ for $0\le i\le k^v$ such that
  \begin{enumerate}[(a)]
  \item $q_0^v\in\iota.w_0^v$, $q_{i+1}^v\in q_i^v.w_{i+1}^v$ for
    $1\le i<k^v$, and $q_{k^v}^v\in G$
  \item $v_0^v=v$ and, for $0\le i<k^v$, $v_{i+1}^v=
    \begin{cases}
      \proc(v_i^v) & \text{ if }w_{i+1}^v\in\proc[n]^*\\
      \msg(v_i^v) & \text{ if }w_{i+1}^v\in\msg[n]^*
    \end{cases}$
  \end{enumerate}

  We define inductively a sequence of subsets of $V$: Let
  $H_0=\emptyset$. Inductively, let $H_{n+1}\subseteq
  V\setminus\bigcup_{0\le i\le n}H_i$ be nonempty and finite such that
  \begin{enumerate}[(A)]
  \item $\bigcup_{0\le i\le n+1}H_i$ is downwards closed in $M$,
  \item for any $v\in V\setminus\bigcup_{0\le i\le n+1}H_i$ with
    $\lambda(v)=p?q$,
    \begin{enumerate}[(B1)]
    \item there exist infinitely many $v'\in V\setminus\bigcup_{0\le
        i\le n+1}H_i$ with $\lambda(v)=\lambda(v')$,
    \item there exist $u,u'\in H_{n+1}$ with $(u,u')\in\msg$ and
      $\lambda(u')=\lambda(v)$,
    \end{enumerate}
  \item for any $v\in H_n$ with $c(v)=1$, we have $v_{k^v}^v\in
    H_n \cup H_{n+1}$. 
  \end{enumerate}
  Then $V=\bigcup_{n\ge0}H_n$. 

  Now set, for $v\in H_n$,
  \begin{enumerate}[$\bullet$]
  \item $c_0(v)=n\bmod 2$ and $d_v=c_0(v)$
  \item if $v$ is a send event, then\cb{
    \begin{align*}
      A_v &= \{q_i^{\bar v}\mid \bar v\in H_{n-1}, c(\bar v)=1, 0\le i<k^{\bar
        v}\text{ such that }v=v_i^{\bar v}\text{ and }w_{i+1}^{\bar
        v}\in\proc[n]^*\}\\
      B_v &= \{q_i^{\bar v}\mid \bar v\in H_{n}, c(\bar v)=1, 0\le i<k^{\bar
        v}\text{ such that }v=v_i^{\bar v}\text{ and }w_{i+1}^{\bar
        v}\in\proc[n]^*\}\\
      A'_v &= \{q_i^{\bar v}\mid \bar v\in H_{n-1}, c(\bar v)=1, 0\le i<k^{\bar
        v}\text{ such that }v=v_i^{\bar v}\text{ and }w_{i+1}^{\bar
        v}\in\msg[n]^*\}\\
      B'_v &= \{q_i^{\bar v}\mid \bar v\in H_{n}, c(\bar v)=1, 0\le i<k^{\bar
        v}\text{ such that }v=v_i^{\bar v}\text{ and }w_{i+1}^{\bar
        v}\in\msg[n]^*\}\\
      d_v' &= c_0(v)
    \end{align*}
  \item if $v$ is a receive event, then
    \begin{align*}
      A_v &= \{q_i^{\bar v}\mid \bar v\in H_{n-1}, c(\bar v)=1, 0\le i<k^{\bar
        v}\text{ such that }v=v_i^{\bar
        v}\}\\
      B_v &= \{q_i^{\bar v}\mid \bar v\in H_n, c(\bar v)=1, 0\le i<k^{\bar v}\text{
        such that }v=v_i^{\bar v}\}\\
      A'_v &= A'_{\msg^{-1}(v)}\\
      B'_v &= B'_{\msg^{-1}(v)}\\
      d_v'&= d'_{\msg^{-1}(v)}
    \end{align*}
    }
  \end{enumerate}
  Then the pair of mappings $(\rho,\mu)$ with $\rho(v)=(A_v,B_v,d_v)$
  and \cb{$\mu(v)=(A'_v,B'_v,d'_v)$} is an accepting run of the CFM on
  $(M,c_0,c_1,\dots,c_n,c)$ and $(M,c_0)\in\Col$. 
\end{proof}

\begin{prop}\label{P-any-1-justified}
  Let $I\subseteq\Ch$. There is a CFM~$\cA_1$ that accepts $(M,c_1,\dots,c_n,c)$
  with $\Inf(M)=I$ iff
  \[
  \forall v\in V:c(v)=1~\Longrightarrow~ \exists W\in L_\pi:
  (M,c_1,\dots,c_n),v\models W. 
  \]
  The number of local states per process is in~$2^{O(|\cP|+s(\pi))}$
  and the number of messages is in~$2^{O(s(\pi))}$. 
\end{prop}

\begin{proof}
  By Lemma~\ref{L-intersection-index=1}, there exists a CFM $\cB$ with
  the given number of states and messages that accepts
  $(M,c_0,c_1,\dots,c_n,c)$ iff it is accepted by~$\cA_\Col$ from
  Prop.~\ref{P-Col-regular} and by the above CFM~$\cA$, i.e., iff
  $(M,c_0)\in\Col$, and $(M,c_0,c_1,\dots,c_n,c)$ is accepted by
  $\cA$. Projecting away the function~$c_0$ gives the CFM~$\cA_1$ by
  the above two lemmas. 
\end{proof}

\noindent \emph{Proof}~(of Theorem~\ref{T-pi-alpha}). The result follows
immediately from Propositions~\ref{P-any-0-justified},
\ref{P-any-1-justified}, and Lemma~\ref{L-intersection-index=1}.\qed

\subsection{The overall construction}

\begin{thm}\label{T-local-formulas}
  Let $I\subseteq\Ch$ and let $\alpha$ be a local formula of \DLTL. 
  Then one can construct a CFM~$\cB$ of index~$1$ such that
  $(M,(c_\beta)_{\beta\in\sub(\alpha)})$ with $\Inf(M)=I$ is accepted
  by $\cB$ iff $c_\beta:V\to\{0,1\}$ is the characteristic function of
  the set of positions that satisfy~$\beta$ for all
  $\beta\in\sub(\alpha)$. 

  With $m$ the number of subformulas of the form
  $\left<\pi\right>\gamma$ and $\left<\pi\right>^{-1}\gamma$ and
  $n\in\N$ such that $s(\pi;\gamma)\le n$ for all such subformulas,
  the number of local states per process is in $2^{O(m(n+|\cP|))}$ and
  the number of control messages is in $2^{O(mn)}$. 
\end{thm}

\begin{proof}
  The CFM $\cB$ has to accept $(M,(c_\beta)_{\beta\in\sub(\alpha)})$
  iff
  \begin{enumerate}[(1)]
  \item $\cA_\sigma$ accepts $(M,c_\sigma)$ for all
    $\sigma\in\sub(\alpha)\cap\Sigma$ (cf.\
    Example~\ref{E-atomic-automata}),
  \item $\cA_\lor$ accepts
    $(M,c_\gamma,c_\delta,c_{\gamma\lor\delta})$ for all
    $\gamma\lor\delta\in\sub(\alpha)$ (cf.\
    Example~\ref{E-BC-of-automata}),
  \item $\cA_\neg$ accepts $(M,c_\gamma,c_{\neg\gamma})$ for all
    $\neg\gamma\in\sub(\alpha)$ (cf.\ Example~\ref{E-BC-of-automata}),
  \item $\cA_{\left<\pi\right>\gamma}$ accepts
    $(M,c_{\alpha_1},\dots,c_{\alpha_n},c_{\gamma},c_{\left<\pi\right>\gamma})$
    for all $\left<\pi\right>\gamma\in\sub(\alpha)$ where
    $\alpha_1,\dots,\alpha_n$ are those local formulas for which
    $\{\alpha_i\}$ appears in the path expression $\pi$ (cf.\
    Theorem~\ref{T-pi-alpha}), and
  \item $\cA_{\left<\pi\right>^{-1}\gamma}$ accepts
    $(M,c_{\alpha_1},\dots,c_{\alpha_n},c_{\gamma},c_{\left<\pi\right>^{-1}\gamma})$
    for all $\left<\pi\right>^{-1}\gamma\in\sub(\alpha)$ where
    $\alpha_1,\dots,\alpha_n$ are those local formulas for which
    $\{\alpha_i\}$ appears in the path expression $\pi$ (cf.\
    Theorem~\ref{T-pi-back-alpha}). 
  \end{enumerate}
  Recall that the CFMs from~(4) all have index~$1$, their number of
  local states per process is bounded by $2^{O(n+|\cP|)}$, and their
  number of messages is bounded by $2^{O(n)}$. Hence, by
  Lemma~\ref{L-intersection-index=1}, there exists a CFM of index $1$
  that checks all the requirements in (4). Its number of states is in
  \begin{align*}
    (m+1)\cdot \prod_{\left<\pi\right>\gamma\in\sub(\alpha)}
    2^{O(n+|\cP|)} &\subseteq (m+1)\cdot
    2^{O(m(n+|\cP|))}\\
    &\subseteq 2^{O(m(n+|\cP|))}
  \end{align*}
  and the number of control messages belongs to
  \[
  \prod_{\left<\pi\right>\gamma\in\sub(\alpha)} 2^{O(s(\pi))}\subseteq
  2^{O(mn)}\ . 
  \]
  Any tuple of local states in any of the CFMs from (5) is accepting. 
  Furthermore, any of them has $2^{O(n)}$ local states per process and
  equally many messages. Hence there is a CFM with $2^{O(mn)}$ local
  states per process and equally many messages that checks all the
  requirements in (5). Furthermore, all tuples of states of this
  machine are accepting. 

  Recall that the CFMs $\cA_\sigma$, $\cA_\lor$, and $\cA_\neg$ have
  just one local state per process, i.e., they only restrict the
  labels $(\sigma,(b_\beta)_{\beta\in\sub(\alpha)})$ allowed in~$M$. 
  Hence, without additional states or messages, one can change the
  above two CFMs into a CFM $\cB$ of index~$1$ that checks (1)--(5). Its
  number of local states per process is in $2^{O(m(n+|\cP|))}$ and its
  number of messages in $2^{O(mn)}$. 
\end{proof}

\section{Translation of global formulas}\label{sec:globform}
A \emph{basic global formula} is a formula of the form $\rA\alpha$ or
$\rE\alpha$ for $\alpha$ a local formula. Then global formulas are
positive Boolean combinations of basic global formulas. 

\begin{prop}\label{P-global-formulas}
  Let $\varphi$ be a global formula and $I\subseteq\Ch$. Then one can
  construct a CFM~$\cA$  that accepts $M$ with $\Inf(M)=I$ iff
  $M\models\varphi$. 

  With $\ell$ the number of basic global subformulas of~$\varphi$, $m$ the
  number of subformulas of the form $\left<\pi\right>\beta$ and
  $\left<\pi\right>^{-1}\beta$, and $n\in\N$ such that $s(\pi;\beta)\le n$ for
  all such subformulas, the number of local states per process is in
  $2^{O(m(n+|\cP|))+|\cP|\ell}$, the number of control messages is in
  $2^{O(\ell+mn)}$, and the index is at most $|\cP^\ell|$. 
\end{prop}

\begin{proof}
  Let~$H$ be the set of basic global subformulas of~$\varphi$. Let
  $\beta=\bigwedge\{\alpha\mid \rE\alpha\in H\text{ or }\rA\alpha\in
  H\}$. Using Proposition~\ref{P-intersection-index-large}, one can
  construct a CFM that accepts $(M,(c_\gamma)_{\gamma\in\sub(\beta)})$
  with $\Inf(M)=I$ iff
  \begin{enumerate}[$\bullet$]
  \item $c_\gamma$ is the characteristic function of the set of
    positions satisfying $\gamma$ for all $\gamma\in\sub(\beta)$
    (Thm.~\ref{T-local-formulas})
  \item $M\models\rA\alpha$ for all $\rA\alpha\in H$
    (Example~\ref{E-E-and-A})
  \item $M\models\rE\alpha$ for all $\rE\alpha\in H$
    (Example~\ref{E-E-and-A}). 
  \end{enumerate}
  Recall that the CFM checking $c_\gamma$ as well as those checking
  $\rA\alpha$ all have index~$1$ while the CFM for $\rE\alpha$ have
  index~$|\cP|$. Hence the number of local states per process of the
  resulting CFM belongs to $1+(|H|+1)\cdot 2^{O(m(n+|\cP|))} \cdot
  2^{|H|}\cdot |\cP|^{|H|} \subseteq 2^{O(m(n+|\cP|))+|\cP|\ell}$, its
  number of messages is in $2^{O(mn)}$, and its index is at most
  $|\cP^H|$. Let $\cA_H$ denote the projection of this CFM to the set
  of MSCs (i.e., we project away the labelings $c_\gamma$). Then
  $\cA_H$ accepts an MSC $M$ with $\Inf(M)=I$ iff $M\models\psi$ for
  all $\psi\in H$. 

  Now the CFM $\cA$ is the disjoint union of at most $2^\ell$ many
  CFMs of the form $\cA_H$. 
\end{proof}

\begin{thm}\label{T-global-formulas}
  Let $\varphi$ be a global formula of PDL. Then one can construct a CFM~$\cA$
  that accepts $M$ iff $M\models\varphi$. 

  With $\ell$ the number of basic global subformulas of~$\varphi$, $m$
  the number of subformulas of the form $\left<\pi\right>\beta$, and
  $n\in\N$ such that $s(\pi;\beta)\le n$ for all such subformulas, the
  number of local states per process is in
  $2^{O(m(n+|\cP|)+|\cP|\ell+|\cP|^2)}$ and the number of control
  messages is in $2^{O(\ell+mn+|\cP|^2)}$. 
\end{thm}

\begin{proof}
  Let, for $I\subseteq\Ch$, $\cA_I$ denote the CFM from
  Prop.~\ref{P-global-formulas} and $\cB_I$ that from
  Prop.~\ref{P-(in)finite-channels}. Using
  Prop.~\ref{P-intersection-index-large}, one can construct a
  CFM~$\Cc_I$ accepting $L(\cA_I)\cap L(\cB_I)$. The number of local
  states per process of this CFM is $3\cdot 2^{O(|\cP|)}\cdot
  2^{O(m(n+|\cP|)+|\cP|\ell)}\cdot|\cP^H|$. 

  Then the disjoint union $\cA$ of all these CFMs $\Cc_I$ for
  $I\subseteq\Ch$ has all the desired properties. 
\end{proof}

\section{Model checking}\label{sec:modelcheck}

\subsection{CFMs vs.\ PDL specifications}

We aim at an algorithm that decides whether, given a global formula
$\varphi$ and a CFM $\cA$, every MSC $M\in L(\cA)$ satisfies
$\varphi$. The undecidability of this problem can be shown following,
e.g., the proof in \cite{MenR04} (that paper deals with Lamport
diagrams and a fragment LD$_0$ of \DLTL, but the proof ideas can be
easily transferred to our setting). To gain decidability, we follow
the successful approach of, e.g., \cite{MadM01,GenMSZ02,GenKM06}, and
restrict attention to existentially $B$-bounded MSCs from~$L(\cA)$.

For a finite or infinite word $w\in\Sigma^\infty$ and $a\in\Sigma$,
let~$|w|_a$ denote the number of occurrences of $a$ in $w$. For $0\le
i\le j<|w|$, the infix $w[i,j]$ is the factor of $w$ starting in
position~$i$ and ending in position~$j$, i.e., $w=u\,w[i,j]\,v$ with
$|u|=i$ and $|w[i,j]|=j-i+1$. If $|w| >i$, then we write $w(i)$ for
$w[i,i]$, the letter no.~$i+1$ in $w$ (note that~$w(0)$ is the first
letter of~$w$).

Let $M=(V,\le,\lambda)$ be an MSC. A \emph{linearization} of $M$ is a
linear order ${\preceq}\supseteq{\le}$ on $V$ of order type at most
$\omega$ (i.e., also with respect to $\preceq$, any node $v\in V$
dominates a finite set). Since equally-labeled nodes of $M$ are
comparable, we can safely identify a linearization of $M$ with a word
from~$\Sigma^\infty$.

A word $w\in\Sigma^\infty$ is \emph{$B$-bounded} (where $B\in\N$) if,
for any $(p,q)\in\Ch$ and any finite prefix $u$ of $w$,
$0\le|u|_{p!q}-|u|_{q?p}\le B$. An MSC $M$ is \emph{existentially
  $B$-bounded} if it admits a $B$-bounded linearization. Intuitively,
this means that the MSC $M$ can be scheduled in such a way that none
of the channels $(p,q)$ ever contains more than $B$ pending messages.

\begin{lem}
  A $B$-bounded word $w\in\Sigma^\infty$ is a linearization of some
  MSC $M$ iff, for any $(p,q)\in\Ch$, any finite prefix of $w$ can be
  extended to a finite prefix $u$ of $w$ such that
  \begin{enumerate}[(1)]
  \item $|u|_{p!q}=|u|_{q?p}$ or
  \item the last letter of $u$ is $p!q$.
  \end{enumerate}
\end{lem}

\begin{proof}
  First suppose that $w$ is a linearization of some MSC. Then
  $|w|_{p!q}=|w|_{q?p}$. If this number is finite, we can extend any
  finite prefix to some finite prefix satisfying (1). Otherwise, any
  suffix contains at least one occurrence of $p!q$, so any prefix can
  be extended to some larger prefix ending with $p!q$.

  Conversely suppose that any finite prefix can be extended to a
  finite prefix satisfying (1) or~(2). We construct from $w$ an MSC as
  follows:
  \begin{enumerate}[$\bullet$]
  \item the set of nodes equals $V=\{v\in\N\mid v<|w|\}$,
  \item for $v\in V$ let $\lambda(v)=w(v)$,
  \item let $(i,j)\in\proc'$ iff $0\le i<j<|w|$ and there exists a
    process $p\in\cP$ with $\lambda(i),\lambda(j)\in\Sigma_p$ and, for
    all $k$ with $i\le k<j$ and $\lambda(k)\in \Sigma_p$, we have
    $i=k$,
  \item let $(i,j)\in\msg'$ iff $i,j\in V$ and there exists a channel
    $(p,q)\in\Ch$ such that $w(i)=p!q$, $w(j)=q?p$, and
    $|w[0,i]|_{p!q}=|w[0,j]|_{q?p}$,
  \item then set ${\preceq}=(\msg'\cup\proc')^*\subseteq V^2$.
  \end{enumerate}
  Suppose $(i,j)\in\msg'$ and $j<i$. Then
  $|w[0,j]|_{p!q}-|w[0,j]|_{q?p}<|w[0,i]|_{p!q}-|w[0,j]|_{q?p}=0$,
  contradicting the $B$-boundedness of $w$. Hence $\msg'$ and $\proc'$
  are contained in $\le$ proving that $\preceq$ is a partial order
  on~$V$. Since $\preceq$ is contained in the natural order $\le$ on
  the set of natural numbers~$V$, the word $w$ is a linearization of
  $M=(V,\preceq,\lambda)$. It therefore remains to be shown that $M$
  is an MSC:
  \begin{enumerate}[$\bullet$]
  \item It is easily verified that $\msg=\msg'$ and $\proc=\proc'$
    implying ${\preceq}=(\msg\cup\proc)^*$.
  \item By the definition of $\proc'$, any two nodes $i$ and $j$ with
    $P(i)=P(j)$ are ordered by~$\preceq$.
  \item Let $(p,q)\in\Ch$ be some channel. Since $w$ is $B$-bounded,
    we have $|w|_{p!q}\ge|w|_{q?p}$. Now suppose
    $|w|_{p!q}>|w|_{q?p}$. Then there are only finitely many
    occurrences of $q?p$; let $u_1$ with $|u_1|_{p!q}-|u_1|_{q?p}>0$
    be a finite prefix of $w$ that contains all occurrences of~$q?p$.
    Then by our assumption on $w$, we can extend $u_1$ to a finite
    prefix $u_2$ of $w$ whose last letter is $p!q$. Hence
    $|u_2|_{p!q}-|u_2|_{q?p}>|u_1|_{p!q}-|u_1|_{q?p}$.  Inductively,
    we find a finite prefix $u$ with $|u|_{p!q}-|u|_{q?q}>B$,
    contradicting the $B$-boundedness of $w$.  Hence
    $|\lambda^{-1}(p!q)|=|w|_{p!q}=|w|_{q?p}=|\lambda^{-1}(q?p)|$
    which finishes the proof that $(V,\preceq,\lambda)$ is an MSC.
  \end{enumerate}
\end{proof}

We next construct, from a CFM $\cA=(C,(\cA_p)_{p \in \cP},F)$ and a
bound $B\in\N$, a finite transition system over $\Sigma$ with multiple
B\"uchi-acceptance conditions that accepts the set of $B$-bounded
linearizations of MSCs from $L(\cA)$: \cb{
  \begin{enumerate}[$\bullet$]
  \item A configuration is a tuple $((s_p)_{p\in\cP},\chi,(p,q))$ with
    the current local states $s_p\in S_p$ for all $p\in \cP$, the
    channel contents $\chi:\Ch\to C^*$ with $|\chi(p',q')|\le B$ for
    all $(p',q')\in \Ch$, and the last active channel $(p,q)\in\Ch$.
  \item The initial configuration is the tuple
    $((\iota_p)_{p\in\cP},\chi,(p,q))$ with $\chi(p',q')=\varepsilon$
    for all $(p',q')\in\Ch$, and where $(p,q)\in\Ch$ is an arbitrary
    but fixed channel, i.e., the local machines are in their initial
    state and all channels are empty.
  \item We have a transition
    \[
    ((s^1_p)_{p\in\cP},\chi^1,(p^1,q^1)) \xrightarrow{a}
    ((s^2_p)_{p\in\cP},\chi^2,(p^2,q^2))
    \]
    for an action $a\in\Sigma_p$ iff there exists a control message
    $c\in C$ such that
    \begin{enumerate}[(T1)]
    \item $s^1_p\xrightarrow{a,c}_ps^2_p$ is a transition of the local
      machine $\cA_p$ and $s^1_q=s^2_q$ for $q\neq p$.
    \item Send events: if $a=p!q$, then $\chi^2(p,q)=\chi^1(p,q)\, c$
      (i.e., message $c$ is inserted into the channel from $p$ to $q$)
      and $\chi^1(p',q')=\chi^2(p',q')$ for $(p',q')\not=(p,q)$ (i.e.,
      all other channels are unchanged)
    \item Receive events: if $a=p?q$, then $\chi^1(q,p)=c\,
      \chi^2(q,p)$ (i.e., message $c$ is deleted from the channel from
      $q$ to $p$) and $\chi^1(q',p')=\chi^2(q',p')$ for
      $(q',p')\not=(q,p)$ (i.e., all other channels are unchanged)
    \item $(p^2,q^2)$ is the channel that $a$ writes to or reads from.
    \end{enumerate}
  \end{enumerate}
} A finite or infinite path
$((s^i_p)_{p\in\cP},\chi^i,(p^i,q^i))_{0\le i<\varkappa}$ (for some
$\varkappa\in\N\cup\{\omega\}$) in this transition system is
\emph{successful} if
\begin{enumerate}[(S1)]
\item there exists a tuple $(f_p)_{p\in\cP}\in F$ such that, for all
  $p\in\cP$ and $0\le i<\varkappa$, there exists $i\le j<\varkappa$
  with $s_p^j=f_p$ and
\item for all $(p,q)\in\Ch$ and $0\le i<\varkappa$, there exists $i\le
  j<\varkappa$ such that \cb{$\chi^j(p,q)=\varepsilon$} or
  $(p^j,q^j)=(p,q)$.
\end{enumerate}

\begin{lem}\label{seqsystem}
  Let $w\in\Sigma^\infty$. Then the following are equivalent:
  \begin{enumerate}[\em(i)]
  \item $w$ is the label of some successful path in the above
    transition system.
  \item $w$ is a $B$-bounded linearization of some MSC from $L(\cA)$.
  \end{enumerate}
\end{lem}

\begin{proof}
  To prove the implication \textit{(ii)}$\Rightarrow$\textit{(i)}, let
  $M=(V,\preceq,\lambda)\in L(\cA)$ be an MSC accepted by $\cA$, let
  $w\in\Sigma^\infty$ be a $B$-bounded linearization of $M$, and let
  $(\mu,\rho)$ be a successful run of $\cA$ on $M$. Without loss of
  generality, we can assume $V=\{v\in\N\mid 0\le v<|w|\}$ and
  ${\preceq}\subseteq{\le}$ such that $w$ is the sequence of labels of
  $(V,\le,\lambda)$. For $i=0$, let $((s_p^i),\chi^i,(p^i,q^i))$ be
  the initial configuration of the transition system. Now let $i>0$.
  For $p\in\cP$, let $s^i_p=\iota_p$ if there is no $0\le j<i$ with
  $w(j)\in\Sigma_p$; otherwise set $s^i_p=\rho(j)$ for $j$ the maximal
  natural number with $j<i$ and $w(j)\in\Sigma_p$. For $(p,q)\in\Ch$,
  set $\chi^i(p,q)=\mu(j_1)\,\mu(j_2)\dots\mu(j_k)$ where $0\le
  j_1<j_2<\dots<j_k<i$ is the sequence of those nodes from $V$ with
  $\lambda(j_\ell)=p!q$ and $\msg(j_\ell)\ge i$ (since $w$ is
  $B$-bounded, we have $0\le k\le B$). Finally, $(p^i,q^i)$ is the
  channel that the action $w(i-1)$ writes to or reads from. Then it
  can be checked that the sequence of these configurations
  $((s_p^i),\chi^i,(p^i,q^i))_{0\le i<|w|}$ forms a $w$-labeled path
  in the transition system. We show that it is successful:
  \begin{enumerate}[(S1)]
  \item Since $(\rho,\mu)$ is successful, there exists
    $(f_p)_{p\in\cP}\in F$ such that for all $p\in\cP$ and all $v\in
    V$ with $\lambda(v)\in\Sigma_p$, there exists $v'\in V$ with
    $\lambda(v')\in\Sigma_p$, $v\preceq v'$, and $\rho(v')=f_p$ (or
    $f_p=\iota_p$ if no such node $v$ exists). Now let $0\le i<|w|$
    and let $v<i$ denote the maximal natural number with
    $w(v)\in\Sigma_p$ (the case that no such number exists is left to
    the reader). Then there exists $v'\in V$ with
    $\lambda(v')\in\Sigma_p$, $v\preceq v'$, and $\rho(v')=f_p$.
    Because of the maximality of $v$, we obtain $i< v'$. Furthermore,
    $s_p^{v'+1}=\rho(v')=f_p$.
  \item Let $0\le i<|w|$. Since $w$ is $B$-bounded, the previous lemma
    implies the existence of $i\le j<|w|$ such that
    $|w[0,j]|_{p!q}=|w[0,j]|_{q?p}$ or the last letter of $w[0,j]$ is
    $p!q$. Hence $\chi^{j+1}(p,q)=\varepsilon$ or
    $(p^{j+1},q^{j+1})=(p,q)$.
  \end{enumerate}

  Conversely assume \textit{(i)}. Since all the channels in the
  transition system contain at most $B$ messages, the word $w$ is
  $B$-bounded. Since the $w$-labeled path satisfies (S2), the word~$w$
  is, by the previous lemma, a linearization of some MSC. Now, using
  (S1) it can be verified similarly to above that this MSC is accepted
  by~$\cA$.
\end{proof}

\begin{thm}\label{T-decision}
  The following problem is PSPACE-complete:\\
  Input: $B\in\N$ (given in unary), CFM $\cB$, and a global formula $\varphi\in\DLTL$.\\
  Question: Is there an existentially $B$-bounded MSC $M\in L(\cB)$
  with $M\models\varphi$?
\end{thm}

\begin{proof}
  Theorem~\ref{T-global-formulas} allows to build a CFM $\cA_\varphi$
  that accepts $M$ iff $M\models\varphi$. From
  Proposition~\ref{P-intersection-index-large}, we then obtain a CFM
  $\cA$ with $L(\cA)=L(\cB)\cap L(\cA_\varphi)$, i.e., $M\in L(\cA)$
  iff $M\in L(\cB)$ and $M\models\varphi$. To decide the existence of
  some existentially $B$-bounded MSC in $L(\cA)$, it suffices to
  decide whether the above transition system has some successful path. 
  Recall that such a path has to simultaneously satisfy
  $b=|\cP|+|\Ch|$ many B\"uchi-conditions. Extending the
  configurations of the transition system by a counter that counts up
  to $b+1$ allows to have just one B\"uchi-condition~\cite{Cho74}. 
  Note that any configuration of the resulting transition system can
  be stored in space
  \[
  \log(b) + |\cP|\log n + |\Ch| B \log|C| + \log|\Ch|
  \]
  where $C$ is the set of message contents of $\cA$ and $n$ is the
  maximal number of local states a process of $\cA$ has. But due to
  Theorem~\ref{T-global-formulas} the size of the CFM $\cA_\varphi$ is
  exponential in the size of $\varphi$. By
  Proposition~\ref{P-intersection-index-large}, $\cA$ stays
  exponential in the size of the input. Hence, the model checking
  problem can be decided in polynomial space. 

  The hardness result follows from PSPACE-hardness of LTL model
  checking. 
\end{proof}

\subsection{HMSCs vs.\ PDL specifications}

In~\cite{Pel00}, Peled provides a $\mathrm{PSPACE}$ model checking
algorithm for high-level message sequence charts (HMSCs) against
formulas of the logic $\mathrm{TLC^-}$. The logic $\mathrm{TLC^-}$ is
a fragment of our logic \DLTL as can be shown easily. Now, we aim to
model check an HMSC against a global formula of \DLTL, and, thereby,
to generalize Peled's result. 

\begin{defi}\label{HMSC}
  An \emph{HMSC} $\cH=(S,\rightarrow, s_0,c,\cM)$ is a finite,
  directed graph $(S,\rightarrow)$ with initial node $s_0\in S$, $\cM$
  a finite set of finite MSCs, and a labeling function $c:S\to\cM$. 
\end{defi}

For defining the semantics of HMSCs we need a concatenation operation. 
Let $M_1=(V_1,\le_1,\lambda_1)$ and $M_2=(V_2,\le_2,\lambda_2)$ be two
finite MSCs over the same process set $\cP$ with disjoint node sets. 
Then $M_1M_2=(V,\le,\lambda)$ is given by $V=V_1\cup V_2$,
$\lambda=\lambda_1\cup\lambda_2$, and $\le$ is the least partial order
containing $\le_1\cup\le_2$ and $\{(v_1,v_2)\mid v_1\in V_1,v_2\in
V_2, P(v_1)=P(v_2)\}$. Informally, the events of $M_2$ succeed the
events of $M_1$ for each process, respectively. 

Let $\cH=(S,\rightarrow, s_0,c,\cM)$ be an HMSC. Let $\eta$ be a
maximal path of $(S,\rightarrow)$ starting in $s_0$, i.e., either a
path $\eta=s_0\rightarrow s_1 \rightarrow\dots\rightarrow s_n$ that
ends in an $s_n\in S$ such that there is no $s\in S$ with
$s_n\rightarrow s$ or an infinite path $\eta=s_0\rightarrow s_1
\rightarrow\dots$. The labeling function $c$ can now be extended to
paths by $c(\eta)=c(s_0)c(s_1)\dots$. The MSC language of the HMSC
$\cH$ is now $L(\cH)=\{c(\eta)\mid \text{$\eta$ is a maximal path
  starting in $s_0$}\}$. Note that for any HMSC $\cH$ the language
$L(\cH)$ is existentially $B$-bounded for some $B\in \N$. Indeed,
since any finite MSC $M$ is existentially $B_M$-bounded for some
$B_M\in\N$, there is a $B$-bounded linearization for every $c(\eta)$
when $B=\max\{B_M\mid M\in\cM\}$. 

\begin{thm}\label{HMSC-decision}
  The following problem is PSPACE-complete:\\
  Input: An HMSC $\cH$ and a global formula $\varphi\in\DLTL$.\\
  Question: Is there an MSC $M\in L(\cH)$ with $M\models\varphi$? 
\end{thm}

\begin{proof}
  \cb{ Let $\cH=(S,\rightarrow,s_0,c,\cM)$ be an HMSC. For every $s\in
    S$ we can find a linearization of the finite MSC $c(s)$. Now, it
    is easy to construct a finite (B\"uchi) automaton $\cS_\cH$ that accepts a
    linearization for each and every MSC $M\in L(\cH)$, and, vice
    versa, each (finite or infinite) word accepted by $\cS_\cH$ is a
    linearization of an $M\in L(\cH)$. Note that the size of $\cS_\cH$
    is linear in the size of $\cH$.

    By Theorem~\ref{T-global-formulas}, we can build a CFM
    $\cA_\varphi$ with $M\in L(\cA_\varphi)$ iff $M\models \varphi$.
    From $\cA_\varphi$ and $\cS_\cH$ (which is implicitly
    existentially $B$-bounded for some $B\in\N$) we construct stepwise
    a transition system $\cS$ by running $\cA_\varphi$ and $\cS_\cH$
    simultaneously (cf.\ the construction before
    Lemma~\ref{seqsystem}). The construction terminates because a run
    of $\cS_\cH$ allows for $B$-bounded linearizations only. A run in
    $\cS$ is successful if both projections of the run are successful.
    Now, $\cS$ admits a successful run iff there is an existentially
    $B$-bounded linearization $w_M$ of some $M\in L(\cH)\cap
    L(\cA_\varphi)$ (where $B$ is implicitly given by $\cH$). An
    analysis similar to the one in the proof of
    Theorem~\ref{T-decision} shows that the existence of a successful
    path of $\cS$ can be decided in polynomial space.

    Again, the hardness result is an easy consequence of PSPACE-hardness of
    LTL model checking. }
\end{proof}

\section{\DLTL with intersection}\label{sec:ipdl}

Several extensions of \DLTL have been considered in the literature that still
come with positive decidability results \cite{Harel2000,GoeLoLu-07}. Though
these results were obtained in the different context of evaluating a formula
over a Kripke structure, it is natural to ask if such extensions can be
handled in our setting as well. We will study here \DLTL with intersection
(\iDLTL, for short), which is the canonical adaption of the logic IPDL, as
defined in \cite{Harel2000}, to our setting. In addition to the local formulas
of \DLTL, we allow local formulas $\left<\pi_1\cap\pi_2\right>\alpha$ where
$\pi_1$ and $\pi_2$ are path expressions and $\alpha$ is a local formula. The
intended meaning is that there exist two paths described by $\pi_1$ and
$\pi_2$ respectively that both lead to the same node $w$ where $\alpha$ holds. 

It is the aim of this section to prove that CFMs are too weak to check all
properties expressed in \iDLTL. To show this result more easily, we also allow
atomic propositions of the form $(a,b)$ with $a,b\in\{0,1\}$; they are
evaluated over an MSC $M=(V,\le,\lambda)$ together with a mapping
$c:V\to\{0,1\}^2$. Then $(M,c),v\models(a,b)$ iff $c(v)=(a,b)$. Let
$\cP=\{0,1\}$ be the set of processes. For $m\ge1$, we first fix an MSC
$M_m=(V_m,\le,\lambda)$ for the remaining arguments: On process $0$, it
executes the sequence $(0!1)^m((0?1)(0!1))^\omega$. The sequence of events on
process $1$ is $(1?0)\,((1?0)\,(1!0))^\omega$. In other words, process $0$
sends $m$ messages to process $1$ and then acknowledges any message received
from $1$ immediately. Differently, process~$1$ has a buffer for two messages. 
After receiving message number $k+1$, it acknowledges message number~$k$. 

Let $E_{0!1}$ denote the set of send-events of process~$0$. For the $k^{th}$
send-event $v$ on process $0$, let $f(v)=((m-k)\bmod
m,(k-1)\mathbin{\mathrm{div}} m)$. Then $f$ maps the set $E_{0!1}$ bijectively
onto the grid $G_m=\{0,1,\dots,m-1\}\times\omega$; we denote the inverse of
$f$ by~$g$. Figure~\ref{fig:foldagrid} shows MSC $M_4$ together with
the mapping~$f$. 

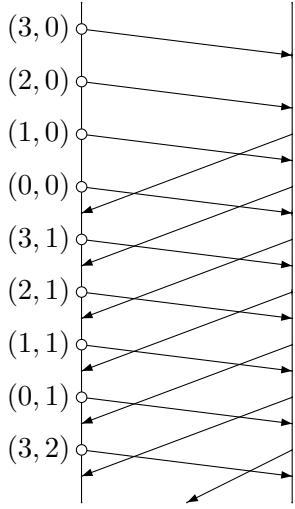
\begin{figure}
  \centering \unitlength .7mm
  \begin{picture}(70,100)(0,-5)

    \gasset{NLangle=180,ExtNL=y,Nw=2,Nh=2,Nfill=n,NLdist=1}
    \node(30)(15,90){$(3,0)$} \node(20)(15,80){$(2,0)$} \node(10)(15,70){$(1,0)$} \node(00)(15,60){$(0,0)$}
    \node(31)(15,50){$(3,1)$} \node(21)(15,40){$(2,1)$} \node(11)(15,30){$(1,1)$} \node(01)(15,20){$(0,1)$}
    \node(32)(15,10){$(3,2)$}

    \gasset{Nw=.1,Nh=.1,AHnb=0}
    \node(0s)(15,95){}\node(0e)(15,0){}
    \drawedge(0s,30){} \drawedge(30,20){} \drawedge(20,10){}
    \drawedge(10,00){} \drawedge(00,31){} \drawedge(31,21){} 
    \drawedge(21,11){} \drawedge(11,01){} \drawedge(01,32){}
    \drawedge(32,0e){}
    \node(1s)(55,95){}\node(1e)(55,0){}\drawedge(1s,1e){}
    \gasset{AHnb=1}
    \node(30r)(55,85){} \node(20r)(55,75){} \node(10r)(55,65){} \node(00r)(55,55){}
    \node(31r)(55,45){} \node(21r)(55,35){} \node(11r)(55,25){} \node(01r)(55,15){}
    \node(32r)(55,5){}

    \drawedge(30,30r){}\drawedge(20,20r){}\drawedge(10,10r){}\drawedge(00,00r){}\drawedge(31,31r){}\drawedge(21,21r){}\drawedge(11,11r){}\drawedge(01,01r){}\drawedge(32,32r){}

    \node(1)(55,70){} \node(2)(55,60){} \node(3)(55,50){} \node(4)(55,40){} \node(5)(55,30){} \node(6)(55,20){} \node(7)(55,10){}
    \node(1r)(15,55){} \node(2r)(15,45){} \node(3r)(15,35){} \node(4r)(15,25){} \node(5r)(15,15){} \node(6r)(15,5){} \node(7r)(35,0){}

    \drawedge(1,1r){} \drawedge(2,2r){} \drawedge(3,3r){} \drawedge(4,4r){} \drawedge(5,5r){} \drawedge(6,6r){} \drawedge(7,7r){}
  \end{picture}
  \caption{MSC $M_4$ and the mapping $f$.} 
  \label{fig:foldagrid}
\end{figure}

\begin{lem}\label{L-iDLTL1}
  There exists a local formula $\alpha$ of \iDLTL such that, for any $m\ge1$
  and any $c:V_m\to\{0,1\}^2$ satisfying $c(g(i,j))=(0,0)$ iff $i=0$, we have
  $(M_m,c)\models\rA\alpha$ iff $c(g(i,j))=c(g(i,j+i))$ for all $(i,j)\in
  G_m$. 
\end{lem}

\begin{proof}
  Let $(i,j)\in G_m$. Then observe the following:
  \begin{enumerate}[$\bullet$]
  \item With $\pi_1$ denoting the path description
    $(\proc;\{(0?1)\})^*;\proc;\{(0!1)\}$, we have that
    $M_m,g(i,j)\models\left<\pi_1\right>\beta$ iff $i>0$ and
    $M_m,g(i-1,j)\models\beta$, or $i=0$ and\linebreak
    ${M_m,g(m-1,j+1)}\models\beta$. 
  \item With $\pi_2$ denoting the path description
    $\msg;\proc;\msg;\proc$, we have
    $M_m,g(i,j)\models\left<\pi_2\right>\beta$ iff
    $M_m,g(i+1,j+1)\models\beta$ whenever $i<m-1$. 
  \end{enumerate}
  As a consequence, we obtain
  \begin{enumerate}[$\bullet$]
  \item if $i>0$, then
    $M_m,g(i,j)\models\left<\pi_1;\pi_2\right>\beta$ iff
    $M_m,g(i,j+1)\models\beta$. 
  \end{enumerate}

  Now let $c:V_m\to\{0,1\}^2$ be a function with $c(g(i,j))=(0,0)$ iff
  $i=0$. Then we have
  \begin{enumerate}[(1)]
  \item
    $(M_m,c),g(i,j)\models\left<\{\neg(0,0)\};(\pi_1;\pi_2)^*\right>\beta$
    iff $i>0$ and there exists $k\ge0$ with
    $(M_m,c),g(i,j+k)\models\beta$,
  \item $(M_m,c),g(i,j)\models
    \left<(\{\neg(0,0)\};\pi_1)^*;\{(0,0)\}\right>\beta$ iff
    $(M_m,c),g(0,j)\models\beta$,
  \item
    $(M_m,c),g(0,j)\models\left<(\pi_2;\{\neg(0,0)\})^*\right>\beta$
    iff there is $0\le k\le m-1$ with
    $(M_m,c),g(k,j+k)\models\beta$. 
  \end{enumerate}
  Now let
  $\pi_3=(\{\neg(0,0)\};\pi_1)^*;\{(0,0)\};(\pi_2;\{\neg(0,0)\})^*$
  and $\pi_4=\{\neg(0,0)\};(\pi_1;\pi_2)^*$. Then, we have
  $(M_m,c),g(i,j)\models\left<\pi_3\cap\pi_4\right>\beta$ iff $i>0$
  and $(M_m,c),g(i,j+i)\models\beta$. Now let
  \[
  \alpha= ((0!1)\land\neg(0,0))\rightarrow \bigwedge_{x\in\{0,1\}^2}
  x\leftrightarrow\left<\pi_3\cap\pi_4\right>x\ . 
  \]
  Then, for all $(i,j)\in G_m$, we have $(M_m,c),g(i,j)\models\alpha$
  iff $c(g(i,j))=c(g(i,j+i))$. 
\end{proof}

\begin{lem}\label{L-iDLTL2}
  Let $\cA=(C,2,(\cA_p)_{p\in\cP},F)$ be a CFM that accepts all
  labeled MSCs $(M_m,c)$ with $m\ge1$ such that
  \begin{enumerate}[\em(1)]
  \item $c(g(i,j))=(0,0)$ iff $i=0$,
  \item $c(g(i,j))=c(g(i,j+i))$ for all $(i,j)\in G_m$. 
  \end{enumerate}
  Then there exist $m\ge1$ and a labeled MSC $(M_m,c)$ accepted by
  $\cA$, satisfying (1), and violating (2). 
\end{lem}

\begin{proof}
  Let $\cA_p=(S_p,\rightarrow_p,\iota_p)$ for $p=0,1$, let $m\ge1$ be
  such that $|S_0|\cdot |S_1|\cdot|C|^{m-1}<3^{\frac{(m-1)(m-2)}{2}}$,
  and let $M_m=(V_m,\le,\lambda)$. Let furthermore $H$ denote the set
  of mappings $c:V_m\to\{0,1\}^2$ satisfying (1), (2), and
  $c(v)=(1,1)$ for all $v\notin E_{0!1}$ (i.e.,
  $\lambda(v)\neq(0!1)$). Then, for any $c\in H$, the pair $(M_m,c)$
  is accepted by $\cA$ -- let $(\rho_c,\mu_c)$ be an accepting run of
  $\cA$ on $(M_m,c)$. 

  Let $v=g(m-1,m-2)$ and $W=\{w\in V\mid w\le v\}$. Then, for any
  event $w\in W$ with $\lambda(w)=1!0$, we have $\msg(w)\in W$. On the
  other hand, there are precisely $m-1$ events $w_1,\dots,w_{m-1}\in
  W$ with $\lambda(w_i)=0!1$ and $\msg(w_i)\notin W$. Let furthermore
  $u\in W$ be the maximal event from process~$1$. 

  Consider $(M_m,c_1)$ and $(M_m,c_2)$ with $c_1,c_2\in H$ and
  $c_1(g(i,j))=c_2(g(i,j))$ for all $0\le j<i<m$. Then $c_1=c_2$
  by~(2). Hence $|H|$ is the number of mappings from $\{(i,j)\mid 0\le
  j<i<m\}$ to $\{0,1\}^2\setminus\{(0,0)\}$, i.e.,
  $3^{\frac{(m-1)(m-2)}{2}}$. 

  Since this number exceeds $|S_0|\cdot|S_1|\cdot|C|^{m-1}$, there
  exist $c_1$ and $c_2$ \cb{with $c_1\neq c_2$} in $H$ with $\rho_{c_1}(v)=\rho_{c_2}(v)$,
  $\rho_{c_1}(u)=\rho_{c_2}(u)$, and $\mu_{c_1}(w_i)=\mu_{c_2}(w_i)$
  for all $1\le i\le m-1$. 

  Now define a mapping $c:V\to\{0,1\}^2$ by $c(x)=c_1(x)$ for $x\in W$
  and $c(x)=c_2(x)$ for $x\notin W$. Then, $c$ satisfies (1) and
  violates (2). But $(M_m,c)$ is accepted by $\cA$: An accepting run
  $(\rho,\mu)$ is defined (similarly to $c$) by
  \[
  \rho(x)=
  \begin{cases}
    \rho_{c_1}(x) & \text{ for }x\in W\\
    \rho_{c_2}(x) & \text{ otherwise}
  \end{cases}
  \text{ and } \mu(x)=
  \begin{cases}
    \mu_{c_1}(x) & \text{ for }x\in W\\
    \mu_{c_2}(x) & \text{ otherwise.} 
  \end{cases}
  \]
\end{proof}

\begin{thm}
  There exists a local formula $\alpha$ of \iDLTL such that the set of
  MSCs $M$ satisfying $\rA\alpha$ cannot be accepted by a CFM. 
\end{thm}

\begin{proof}
  Let $\alpha$ be the local formula from Lemma~\ref{L-iDLTL1}. Towards a
  contradiction, assume $\cA$ is a CFM such that, for any pair $(M,c)$, we
  have $(M,c)\models\rA\alpha$ iff $(M,c)$ is accepted by $\cA$. In
  particular, $\cA$ accepts all pairs $(M_m,c)$ satisfying (1) and (2) from
  Lemma~\ref{L-iDLTL2}. Hence there exists some pair $(M_m,c)$ that is
  accepted by $\cA$, satisfies (1), and violates (2). But now, by
  Lemma~\ref{L-iDLTL1} again, $(M_m,c)\models\neg\rA\alpha$, contradicting our
  assumption on $\cA$. 

  Using a new process $2$, one can encode the mapping $c$ by
  additional messages from processes $0$ and $1$ to process $2$. 
\end{proof}

\section{Open questions}
The semantics of every \DLTL formula $\varphi$ is the behavior of a
CFM~$\cA$.  Hence any \DLTL formula is equivalent to some formula from
existential monadic second order, but a precise description of the
expressive power of \DLTL is not known. Because of quantification over
paths, it cannot be captured by first-order logic
\cite[Prop.~14]{DieG04}. On the other hand, \DLTL is closed under
negation, hence \DLTL is a proper fragment of existential monadic
second order logic. \cb{But it is not even clear that semantical membership of this fragment is decidable.} 

The decidability of the model checking problem for CFMs against
MSO-formulas was shown in \cite{GenKM06} for existentially $B$-bounded
MSCs. For compositional MSCs (a mechanism for the description of sets
of MSCs that is similar but more general than HMSCs) and MSO, the
decidability of the model checking problem was established
in~\cite{MadM01}.  Since the logic \iDLTL, i.e., \DLTL with
intersection, can be translated effectively into an MSO-formula, the
model checking problem is decidable for \iDLTL. However, the
complexity of $\mathrm{MSO}$ model checking is
non-elementary. Therefore, we would like to know if we can do any
better for \iDLTL. 

In \DLTL, we can express properties of the past and of the future of an event
by taking either a backward- or a forward-path in the graph of the MSC. We are
not allowed to speak about a zig-zag-path where e.g.\ a mixed use of $\proc$
and $\proc^{-1}$ would be possible. It is an open question whether formulas of
such a ``mixed \DLTL'' could be transformed to CFMs. 

\bibliographystyle{alpha} \bibliography{references}

\end{document}